\documentclass[conference]{IEEEtran}
\IEEEoverridecommandlockouts

\usepackage{amsmath,amssymb,amsfonts,amsthm}
\usepackage{textcomp}
\usepackage{xcolor}
\usepackage{cite}
\usepackage{bbm}
\usepackage{bm}
\usepackage[hidelinks]{hyperref}
\usepackage{graphicx}

\usepackage{algorithm}
\usepackage{algpseudocode} 

\newtheorem{theorem}{Theorem}
\newtheorem{proposition}[theorem]{Proposition} 

\newtheorem{remark}{Remark}

\newcommand{\R}{\mathbb{R}}
\newcommand{\E}{\mathbb{E}}
\newcommand{\bI}{\bm{I}}
\newcommand{\Yt}{Y_t}
\newcommand{\sYt}{s_{\Yt}}
\newcommand{\sYtT}{s_{\Yt|T}}

\begin{document}

\title{Information Gradient for Nonlinear Gaussian Channel with Applications to Task-Oriented Communication}

\author{
\IEEEauthorblockN{Tadashi Wadayama}
\IEEEauthorblockA{Nagoya Institute of Technology\\
\texttt{wadayama@nitech.ac.jp}}
}

\maketitle

\begin{abstract}
We propose a gradient-based framework for optimizing parametric nonlinear Gaussian channels 
via mutual information maximization. Leveraging the score-to-Fisher bridge (SFB) methodology, 
we derive a computationally tractable formula for the information gradient that is the gradient 
of mutual information with respect to the parameters of the nonlinear front-end. 
Our formula expresses this gradient in terms of two key components: 
the score function of the marginal output distribution, 
which can be learned via denoising score matching (DSM), 
and the Jacobian of the front-end function, which is handled efficiently 
using the vector-Jacobian product (VJP) within automatic differentiation frameworks. 
This enables practical parameter optimization through gradient ascent.
Furthermore, we extend this framework to task-oriented scenarios, 
deriving gradients for both task-specific mutual information, where a task variable depends 
on the channel input, and the information bottleneck (IB) objective. 
A key advantage of our approach is that it facilitates end-to-end optimization 
of the nonlinear front-end without requiring explicit computation on the output distribution.
Extensive experimental validation confirms the correctness of our information gradient formula against 
analytical solutions and demonstrates its effectiveness in optimizing both linear and nonlinear 
channels toward their objectives.
\end{abstract}


\section{Introduction}
\label{sec:introduction}

Modern communication and sensing systems increasingly incorporate components 
that exhibit significant nonlinearities, fundamentally impacting overall performance. 
Examples range from physical hardware limitations, such as radio frequency (RF) amplifier saturation 
in transmitters in wireless communication systems, 
which exhibit nonlinear gain characteristics at high input levels, 
to inherent distortions introduced during signal conversion, 
like those from quantization and analog-to-digital (ADC) and digital-to-analog (DAC) conversions. 
Furthermore, sophisticated signal processing blocks, 
including learned encoders in end-to-end communication systems for semantic communication 
or complex algorithms for sensor fusion, often involve nonlinear transformations.

These nonlinearities are frequently captured by a parametric model for the communication channel:
\begin{align}
Y_t = f_{\bm \eta}(X) + Z_t, \quad Z_t \sim \mathcal{N}(0, t \bm I_m), \label{eq:channel_model}
\end{align}
where $X \in \R^n$ is the input signal, $f_{\bm \eta}: \R^n \to \R^m$ represents 
the deterministic {\em nonlinear front-end} parameterized by 
$\bm \eta \in \R^d$, and $Z_t$ is additive white Gaussian noise (AWGN) with variance $t > 0$.

A fundamental goal in designing and operating such systems is to maximize 
the flow of information from the input $X$ to the output $Y_t$. 
This translates to the optimization problem:
\begin{align}
\label{eq:optimization_problem}
\text{maximize}_{\bm \eta \in \mathcal{F}}\ I(X; Y_t),
\end{align}
where $I$ denotes the {\em mutual information} (MI) and $\mathcal{F}$ 
is a feasible set for the parameters $\bm \eta$. 
Optimizing MI is crucial, as maximizing the mutual information 
with respect to the parameters of the nonlinear front-end 
is fundamental to achieving high data rates and robust performance, 
pushing towards the channel's limits under nonlinear distortions. 
Moreover, in emerging fields such as semantic communication, 
task-oriented sensing, and intelligent wireless systems, 
the ability to optimize these nonlinear transformations via MI maximization 
provides a principled route towards end-to-end learning and performance gains.

However, solving \eqref{eq:optimization_problem} presents significant challenges. 
Evaluating $I(X; \Yt)$ itself is often intractable for general nonlinear $f_{\bm \eta}$, 
posing considerable difficulty especially in high-dimensional settings.
More critically for optimization, 
obtaining a usable gradient of the mutual information with respect 
to the parameters $\bm \eta$ is known to be difficult. 
Conventional methods, such as kernel density estimation-based methods \cite{moon1995}, 
typically require access to the output distribution $p(Y_t)$, 
which are computationally demanding to approximate in high-dimensional settings.

This paper overcomes this hurdle by proposing a novel, computationally 
tractable framework for calculating and utilizing the {\em information gradient}, 
$\nabla_{\bm \eta} I(X; Y_t)$. Leveraging the score-to-Fisher bridge (SFB) methodology \cite{wadayama2025}, 
we derive an analytical expression for this gradient  that depends only on the 
marginal score function $s_{Y_t}(\bm y) = \nabla_{\bm y} \log p_{Y_t}(\bm y)$ of the channel output.
Crucially, this score function can be effectively learned from data samples using 
denoising score matching (DSM) learning methods \cite{hyvarinen2005,vincent2011,song2019,song2021}.
The gradient computation itself is handled efficiently via the vector-Jacobian product (VJP) identity
within standard automatic differentiation (AD) 
frameworks \cite{Baydin2018}, bypassing the need to explicitly form the potentially large Jacobian 
matrix $Df_{\bm \eta}(X)$.

This marginal-free approach enables practical, 
gradient-based optimization of the nonlinear front-end $f_{\bm \eta}$ directly for MI maximization. 
We further extend this framework to handle task-oriented scenarios \cite{Gunduz2023}, 
deriving gradients for maximizing $I(T; Y_t)$, where $T=g(X)$ represents 
a specific task variable derived from the input, as well as for the widely used 
{\em information bottleneck (IB) criterion} \cite{Tishby1999}. 
Our comprehensive experimental validation confirms the validity of 
the proposed gradient formulas against known solutions and demonstrates 
their practical effectiveness in optimizing both linear and nonlinear channel 
models towards their respective objectives.

We have released an implementation, including scripts to reproduce numerical results 
in this paper at: \\ \url{https://github.com/wadayama/information_gradient}.

\section{Preliminaries}
\label{sec:preliminaries}

The SFB methodology~\cite{wadayama2025} provides 
a framework for mutual information estimation in nonlinear Gaussian channels 
via DSM learning. In this section, we review 
the key concepts of the SFB methodology according to \cite{wadayama2025}.

\subsection{Notation}

Unless otherwise stated, all logarithms are natural (nats) in this paper.
Bold letters denote vectors and matrices.
A random variable is denoted by an uppercase letter. 
For a random variable $U$ with density $p_U$, we write $h(U)$ for 
its differential entropy, $I(\cdot;\cdot)$ for mutual information, 
and $\operatorname{Cov}(U)$ for its covariance matrix. 
The Fisher information is defined via the score function 
$s_U(\bm u)\equiv\nabla_{\bm u}\log p_U(\bm u)$ as
\begin{align}
J(U)\equiv \mathbb{E}_{\bm u\sim p_U(\cdot)}\!\left[\|\nabla_{\bm u}\log p_U(\bm u)\|^2\right],
\end{align}
see, e.g.,~\cite{fisher1922,cover2006}.
Other notations used in this paper are summarized in Table \ref{tab:notation}.

\begin{table*}[t]
  \centering
  \footnotesize
  \renewcommand{\arraystretch}{1.2}
  \caption{Notation summary used throughout the paper.}
  \label{tab:notation}
  \begin{tabular}{p{0.17\textwidth} p{0.23\textwidth} p{0.52\textwidth}}
  \hline
  \textbf{Symbol} & \textbf{Type / Domain} & \textbf{Meaning} \\
  \hline
  $X \in \mathbb{R}^n$            & random vector            & Input/source (data) variable. \\
  $W = f_{\eta}(X) \in \mathbb{R}^m$ & random vector          & Encoded/feature variable produced by encoder $f_{\eta}\!:\mathbb{R}^n\!\to\!\mathbb{R}^m$ with parameter $\eta$. \\
  $Z_t \sim \mathcal{N}(0, t \bm I_m)$ & random vector            & Additive Gaussian noise with variance parameter $t \ge 0$. \\
  $Y_t = f_{\bm \eta}(X) + Z_t$       & random vector            & Channel/output variable at noise level $t$. \\
  $s_{Y_t}(\bm y) = \nabla_{\bm y} \log p_{Y_t}(\bm y)$ & vector field    & Marginal score of $Y_t$ (depends on $\bm \eta$ through $p_{Y_t}$). \\
  $s_{\bm \theta}(\bm y;\bm \eta)$             & vector field / NN        & Learned (denoising) score model; optionally conditioned on $\eta$. \\
  $D f_{\bm \eta}(\bm x) \in \mathbb{R}^{m\times n}$ & matrix        & Jacobian of $f_{\bm \eta}$ at $\bm x$; used for VJP/JVP calculus. \\
  $\bm v^{\top} D f_{\bm \eta}(\bm x)$         & row vector               & Vector–Jacobian product (VJP), used to avoid forming full Jacobians. \\
  $\langle \bm u, \bm v\rangle$            & scalar                   & Euclidean inner product.\\
  $J(Y_t)=\mathbb{E}[\|s_{Y_t}(Y_t)\|^2]$ & scalar          & Fisher information of $Y_t$. \\
  $I(U;V)$, $H(U)$                 & scalar                   & Mutual information and (differential) entropy. \\
  $\mathbb{E}[\cdot]$, $\mathrm{Var}[\cdot]$ & operator      & Expectation and variance. \\
  $t \ge 0$                        & scalar                   & Noise/diffusion (variance) parameter. \\
  $\beta \ge 0$                    & scalar                   & IB trade-off parameter in $I(T;Y_t)-\beta I(X;Y_t)$-type objectives. \\
  $\nabla_{\bm x}$                       & operator                 & Gradient w.r.t.\ $\bm x$; similarly $\nabla_{\bm \eta}$ for parameters. \\
  $\operatorname{stop}(\bm v)$         & operator                 & Stop-gradient: treats $\bm v$ as a constant during automatic differentiation. \\
  $L_{\mathrm{DSM}}(\bm \theta)$       & scalar (loss)            & Denoising score matching loss used to train $s_{\bm \theta}$. \\
  $\mathcal{N}(\bm \mu,\bm \Sigma)$        & distribution             & Multivariate normal with mean $\bm \mu$ and covariance $\bm \Sigma$. \\
  \hline
  \end{tabular}
  \end{table*}
  
\subsection{Denoising Score Matching (DSM)}
\label{subsec:dsm}
For the Gaussian-smoothed distribution $p_\sigma \equiv p * \phi_{\sigma^2}$, 
the DSM loss~\cite{vincent2011} is given by
\begin{align}
\mathcal{L}_{\text{DSM}}(\bm \theta) \equiv \E_{\bm{w},\bm{\epsilon}}
\left[\left\|s_{\bm \theta}(\bm{w} + \sigma\bm{\epsilon}) 
+ \frac{\bm{\epsilon}}{\sigma}\right\|^2\right],
\end{align}
where $\bm{w} \sim p$, $\bm{\epsilon} \sim \mathcal{N}(0, \bI)$, 
and $s_{\bm \theta}(\cdot)$ is a parametric score model with the score parameter $\bm \theta$.
The parametric score model is typically implemented as a neural network.
A brief review of DSM learnig suitable for our purpose is given in \cite{wadayama2025}.
Detailed information on DSM learning is given in \cite{vincent2011,song2019,song2021}.

\subsection{Estimation of Fisher Information}
The Fisher information of $\Yt$ is defined by 
$
J(\Yt) \equiv \E_{\Yt}[\|\sYt(Y_t)\|^2],
$
which can be estimated using the learned score model with the Monte Carlo estimator:
\begin{align}
\widehat{J(\Yt)} \equiv \frac{1}{B}\sum_{i=1}^B \|s_{\bm \theta}(\bm{y}_i)\|^2,
\end{align}
where $\bm{y}_i$ is a sample from $\Yt$. Namely, if we have learned score
model $s_{\bm \theta}(\cdot)$, we can estimate the Fisher information of $\Yt$ 
by Monte Carlo estimation.

\subsection{Differential mutual information and Fisher information}

Let $W = f(X)$ be a transformed input.
The de Bruijn identity \cite{stam1959,barron1986entropy}
connects mutual information to Fisher information:
\begin{align}
\frac{d}{dt}I(W; \Yt) = \frac{1}{2}J(\Yt) - \frac{m}{2t}.
\end{align}
By integrating this equation over $t\in[T,\infty)$, 
we have the {\em Fisher integral representation} of MI:
\begin{align}
  I(W;Y_T)=\frac12\int_T^\infty\!\Big(\frac{m}{t}-J(Y_t)\Big)\,dt
\end{align}
that depends only on the Fisher information of $Y_t$.
Furthermore, the following equivalence relation
holds for \emph{any} deterministic function $f$:
\begin{align}
    I(X;Y_t)=I(W;Y_t), \quad Y_t = W + Z_t.
\end{align}
This means that the Fisher integral representation of MI 
can be used to construct an MI estimator:
\begin{align}
    \label{eq:mi-estimation}
     \widehat{I(X;Y_T)} = \frac{1}{2}\int_T^{\infty}\Big(\frac{m}{t}-\widehat J_t\Big) dt,
\end{align}
which is computationally tractable since we only need to estimate 
the Fisher information of $Y_t$ from the learned score function.
This is the key idea of the SFB methodology \cite{wadayama2025}.

\section{Information Gradient}
\label{sec:general}

In this section, we derive the information gradient formula for nonlinear Gaussian channels.
This is the core contribution of this paper.

\subsection{Problem Setup}
\label{subsec:setup}

We focus on the parametric nonlinear channel defined in \eqref{eq:channel_model}. 
To develop our gradient-based optimization framework, we make a few standard assumptions. 
We assume the front-end function $f_{\bm \eta}: \R^n \to \R^m$ is differentiable 
with respect to the parameters $\bm \eta$, which is typically true for functions 
represented by neural networks or other common parametric models. 

We also assume that evaluating the function $f_{\bm \eta}(\bm x)$ for a given input $\bm x$ is 
computationally feasible, allowing for efficient forward simulation of the channel. 
Furthermore, we require that the input distribution $p_X$ is either known analytically 
or can be efficiently sampled, enabling Monte Carlo estimation techniques. 
Finally, throughout our derivations, we assume standard regularity conditions hold, 
permitting the exchange of differentiation and integration operators where necessary. 
These assumptions are generally mild and satisfied in many practical communication 
and signal processing scenarios.

\subsection{Information Gradient Formula}
\label{subsec:gradient}

The information gradient formula is given in the following proposition.
\begin{proposition}[Information gradient]
\label{prop:info_gradient}
Under standard regularity conditions, the gradient of mutual information 
with respect to the front-end parameters $\bm \eta$ is given by
\begin{align}
\nabla_{\bm \eta} I(X; \Yt) = -\E_{X,Z_t}\left[D f_{\bm \eta}(X)^\top \sYt(f_{\bm \eta}(X) + Z_t)\right],
\label{eq:gradient_formula}
\end{align}
where $\sYt(\bm y) \equiv \nabla_y \log p_{\Yt,\eta}(\bm y)$ 
is the score function of the marginal distribution of $\Yt$.
The notation $D f_{\bm \eta}(\bm x) \in \R^{m \times d}$ denotes 
the Jacobian of $f_{\bm \eta}$ with respect to $\bm \eta$.
\end{proposition}

\begin{proof}
Let $W = f_{\bm \eta}(X)$. Since $\Yt$ depends on $X$ only through $W$, the triple $X \to W \to \Yt$ forms a Markov chain. By the data processing inequality and the chain rule for mutual information:
\begin{align}
I(X; \Yt) = I(W; \Yt) + I(X; \Yt|W) = I(W; \Yt),
\label{eq:markov_chain}
\end{align}
where the second equality follows from the conditional independence $X \perp \Yt \mid W$.

The mutual information can be decomposed as
\begin{align}
I(W; \Yt) = h(\Yt) - h(\Yt|W).
\label{eq:mi_decomp_proof}
\end{align}
Given $W = \bm w$, the conditional distribution is $\Yt|W=\bm w \sim \mathcal{N}(\bm w, t\bI_m)$, 
hence we have
\begin{align}
h(\Yt|W=\bm{w}) = \frac{m}{2}\log(2\pi e t),
\end{align}
which is independent of both $\bm w$ and $\bm \eta$. Therefore:
\begin{align}
\nabla_{\bm \eta} h(\Yt|W) = \bm{0}.
\label{eq:cond_entropy_const}
\end{align}
Combining \eqref{eq:markov_chain}, \eqref{eq:mi_decomp_proof}, and \eqref{eq:cond_entropy_const},
we have
\begin{align}
\nabla_{\bm \eta} I(X; \Yt) = \nabla_{\bm \eta} I(W; \Yt) = \nabla_{\bm \eta} h(\Yt).
\label{eq:grad_mi_equals_grad_h}
\end{align}

To compute $\nabla_{\bm \eta} h(\Yt)$, we use the reparameterization $\Yt = f_{\bm \eta}(X) + Z_t$,
yielding 
\begin{align}
h(\Yt) = -\E_{X,Z_t}\left[\log p_{\Yt,\bm \eta}(f_{\bm \eta}(X) + Z_t)\right].
\end{align}

Under standard regularity conditions (differentiability of $f_{\bm \eta}$, 
exchangeability of differentiation and expectation), applying the chain rule yields
\begin{align}
\nabla_{\bm \eta} h(\Yt)
&= -\E_{X,Z_t}\left[\nabla_{\bm \eta} \log p_{\Yt,\bm \eta}(\bm y)\big|_{\bm y=f_{\bm \eta}(X)+Z_t}\right] \nonumber \\
&\quad - \E_{X,Z_t}\left[D f_{\bm \eta}(X)^\top \nabla_{\bm y} \log p_{\Yt,\bm \eta}(\bm y)\big|_{\bm y=f_{\bm \eta}(X)+Z_t}\right].
\label{eq:two_terms_proof}
\end{align}
The first term vanishes by the score property. Specifically, we have 
\begin{align}
\E_{Y_t}\left[\nabla_{\bm \eta} \log p_{\Yt,\bm \eta}(\Yt)\right]
&= \int \nabla_{\bm \eta} p_{\Yt,\bm \eta}(\bm y)\,d\bm y \\
&= \nabla_{\bm \eta} \int p_{\Yt,\bm \eta}(\bm y)\,d\bm y  \\
&= \nabla_{\bm \eta} 1 = \bm{0}.
\label{eq:score_property}
\end{align}
Substituting \eqref{eq:score_property} into \eqref{eq:two_terms_proof},
we have
\begin{align}
\nabla_{\bm \eta} h(\Yt) = -\E_{X,Z_t}\left[D f_{\bm \eta}(X)^\top \sYt(f_{\bm \eta}(X) + Z_t)\right].
\end{align}
\end{proof}

If we have an estimate of the score function $\sYt(\cdot)$, 
we can use this formula to esitmate the information gradient.
This formula is the key to optimizing the parameters $\bm \eta$ to maximize mutual information.

\begin{remark} \rm
Proposition~\ref{prop:info_gradient}
also yields \emph{analytic or semi-analytic} information gradients
whenever the marginal score $s_{Y_t}(\bm y)=\nabla_{\bm y}\log p_{Y_t}(\bm y)$
is available in closed (or near-closed) form. In particular, for
affine-Gaussian channels $f_{\bm \eta}(\bm x)=\bm A\bm x+ \bm b$ with Gaussian-mixture
inputs, $Y_t$ is a finite Gaussian mixture and the score is given by
\begin{align}
s_{Y_t}(\bm y)=\sum_{k=1}^K\gamma_k(\bm y)\,\big(\bm A\bm \Sigma_k \bm A^\top+t \bm I_m\big)^{-1}\!\big((\bm A\bm \mu_k+\bm b)-\bm y\big),
\end{align}
where $\gamma_k(\bm y)\propto \pi_k\,\mathcal N(\bm y;\bm A\bm \mu_k+\bm b,
\bm A\bm \Sigma_k \bm A^\top+t \bm I_m)$.
Plugging this $s_{Y_t}$ into (\ref{eq:gradient_formula}) gives 
$\nabla_{\bm \eta} I(X; \Yt)$ as a closed-form (or near-closed-form) expression.
The same phenomenon occurs for linear-Gaussian models
(recovering the classical log-det gradients), discrete constellations
(PAM/QAM/PSK) under AWGN, and circular/phase models (e.g., PSK,
wrapped Gaussians/von~Mises), where $s_{Y_t}$ admits Bessel-function
expressions. 
\end{remark}

Such analytical cases where the score function is available in closed form are 
interesting topics to discuss, however, our primary focus of the following sections 
is on the general case where the score function is intractable and must be approximated 
via DSM learning.

\subsection{VJP Identity}
\label{subsec:vjp}

Direct computation of the Jacobian $D f_{\bm \eta}(\bm x) \in \R^{m \times d}$ 
with respect to the parameter $\bm \eta$
can be prohibitively expensive for high-dimensional parameter vectors $\bm \eta$, 
as it requires computing all $m \times d$ partial derivatives explicitly. 
We avoid this computational burden by leveraging 
the \emph{Vector-Jacobian Product (VJP)} identity \cite{Baydin2018}.
This subsection provides a detailed explanation of 
the VJP identity and how to use it to compute the information gradient efficiently.

The next proposition states the VJP identity.
\begin{proposition}[VJP identity]
\label{prop:vjp_identity}
Let $f_{\bm\eta}:\R^n\to\R^m$ be $C^1$ in $\bm\eta$. 
Fix $\bm x\in\R^n$ and $\bm v\in\R^m$. 
Then, we have the VJP identity:
\begin{align}
    \label{eq:vjp_identity}
\nabla_{\bm\eta}\,\langle f_{\bm\eta}(\bm x),\bm v\rangle
= D f_{\bm\eta}(\bm x)^{\!\top}\bm v .
\end{align}
\end{proposition}

\begin{proof}
Write $\langle f_{\bm\eta}(\bm x),\bm v\rangle=\sum_{i=1}^m v_i\,f_{\bm\eta,i}(\bm x)$.
For each $j\in\{1,\dots,d\}$, we have
\begin{align}
\frac{\partial}{\partial \eta_j}\,\langle f_{\bm\eta}(\bm x),\bm v\rangle
= \sum_{i=1}^m v_i\,\frac{\partial}{\partial \eta_j} f_{\bm\eta,i}(\bm x).
\end{align}
By the definition $\big[D f_{\bm\eta}(\bm x)\big]_{ij}
\equiv \frac{\partial}{\partial \eta_j} f_{\bm\eta,i}(\bm x)$, the right-hand side equals
$\big((D f_{\bm\eta}(\bm x))^{\!\top}\bm v\big)_j$.
Collecting the components over $j$ yields
$\nabla_{\bm\eta}\,\langle f_{\bm\eta}(\bm x),\bm v\rangle
= \big(D f_{\bm\eta}(\bm x)\big)^{\!\top}\bm v$, as claimed.
\end{proof}

This identity is the key to efficient computation: instead of forming 
the full Jacobian matrix, we compute the product $D f_{\bm \eta}(\bm x)^\top \bm v$ 
directly as the gradient of a scalar function $\langle f_{\bm \eta}(\bm x), \bm v \rangle$ 
with respect to $\bm \eta$. Modern automatic differentiation frameworks implement 
this efficiently through reverse-mode automatic differentiation (AD) \cite{Baydin2018}.

The VJP identity can be used to estimate the information gradient from a mini-batch of samples.
Consider a mini-batch $\{(\bm x_i, \bm z_i)\}_{i=1}^B$ 
where $\bm x_i \sim p_X$ and $\bm z_i \sim \mathcal{N}(\bm{0}, t\bI_m)$.
Let $\bm y_i = f_{\bm \eta}(\bm x_i) + \bm z_i$ denote the corresponding channel output.
The mini-batch-based {\em information gradient estimator} 
using a learned score function $s_{\bm \theta}(\cdot)$ can be written as
\begin{align}
\widehat{\nabla_{\bm \eta} I}
&= -\frac{1}{B}\sum_{i=1}^B D f_{\bm \eta}(\bm x_i)^\top s_{\bm \theta}(f_{\bm \eta}(\bm x_i) + \bm z_i) \nonumber \\
&= -\nabla_{\bm \eta}\left\{\frac{1}{B}\sum_{i=1}^B \langle f_{\bm \eta}(\bm x_i), {\text{stop}(s_{\bm \theta}(f_{\bm \eta}(\bm x_i) + \bm z_i))} \rangle\right\},
\label{eq:batch_vjp}
\end{align}
where the second equality follows from the VJP identity (\ref{eq:vjp_identity}).
Crucially, when computing this gradient with respect to the front-end parameters $\bm \eta$, 
the score function $s_{\bm \theta}(\cdot)$ 
(and its parameters $\bm \theta$) must be treated as fixed.
The $\text{stop}(\cdot)$ operation indicates this `stop-gradient,' 
ensuring that $s_{\bm \theta}(\bm y_i)$ is treated as a constant vector during 
the reverse-mode automatic differentiation.

We define the VJP loss function:
\begin{align}
\mathcal{L}_{\text{vjp}} \equiv \frac{1}{B}\sum_{i=1}^B \langle f_{\bm \eta}(\bm{x}_i), 
\text{stop}(s_{\bm \theta}(\bm{y}_i)) \rangle.
\label{eq:vjp_loss}
\end{align}
The information gradient estimate is thus given by
$
\widehat{\nabla_{\bm \eta} I} = -\nabla_{\bm \eta} \mathcal{L}_{\text{vjp}}.
$

\subsection{Path-Integral Representation of Mutual Information}

Beyond using gradients for optimization, 
our information gradient formula in Proposition \ref{prop:info_gradient} provides a
\emph{path-integral} route 
to obtain the mutual information.
Let $Y_t=f_{\bm\eta}(X)+Z_t$ with 
a differentiable parameter $\bm\eta$ and denote
\begin{align}
G(\bm\eta) \equiv \nabla_{\bm\eta} I(X;Y_t)
=-\mathbb{E}\!\big[D f_{\bm\eta}(X)^\top s_{Y_t}(Y_t)\big].
\end{align}
For any smooth
path $\gamma:[0,1]\!\to\!\mathbb{R}^d$ in the parameter space with $\gamma(0)=\bm\eta_0$ and
$\gamma(1)=\bm\eta_1$, we have the line integral representation of the mutual information:
\begin{equation}
\label{eq:path-int}
I(X;Y_t;\bm\eta_1)-I(X;Y_t;\bm\eta_0)
=\int_0^1 \big\langle G\big(\gamma(s)\big),\dot\gamma(s)\big\rangle\,ds.
\end{equation}
In practice one chooses $\bm\eta_0$ so that $I(X;Y_t;\bm\eta_0)$ is known, 
and evaluates
\eqref{eq:path-int} with numerical integration. 

A crucial advantage of this path-integral approach is its applicability
to estimate the mutual information of the additive non-Gaussian channel,
namely, $Z_t$ can be non-Gaussian.
This is in contrast to the Fisher integral route that requires de Bruijn identity.
However, the path-integral route needs a non-trivial anchor point $\bm\eta_0$ that 
provides a known mutual information value $I(X;Y_t;\bm\eta_0)$.


\begin{remark} \rm 
The total differential of the mutual information 
with respect to the parameter $\bm \eta$ and the noise level 
$t$ is given by
\begin{align}
dI(\bm\eta,t)
=\big\langle\nabla_\eta I(\bm\eta,t),\,d\eta\big\rangle
+\tfrac12\!\left(J(Y_t;\bm\eta)-\tfrac{m}{t}\right)dt
\label{eq:total_differential}
\end{align}
gives a coordinate-wise ``map'' of the mutual-information 
surface $I(\bm\eta,t)$.
The \(\bm \eta\)-term captures how design changes move us 
on the information landscape,
while the \(t\)-term captures how diffusion (noise) deforms it.
Thus Fisher-integral and \(\bm \eta\)-gradient routes are 
two consistent traversals
of the same information landscape.

From the total differential \eqref{eq:total_differential},
we can obtain two complementary integral representations.
\begin{itemize}
\item[(i)] {$t$-direction (Fisher integral route):}
Assuming  the regularity 
to differentiate under the integral sign and 
$\lim_{T\to\infty} I(\eta,T)=0$ , we have 
the Fisher integral representation of MI and its gradient:
\begin{align}
I(\eta,t_0)
&=\frac12\int_{t_0}^{\infty}\!\Big(J(Y_\tau;\eta)-\frac{m}{\tau}\Big)\,d\tau, \\
\nabla_\eta I(\eta,t_0)
&=\frac12\int_{t_0}^{\infty}\!\nabla_\eta J(Y_\tau;\eta)\,d\tau.
\end{align}

\item[(ii)] {Fixed-$t$ path (parameter route):}
For any smooth path $\eta(s)$ joining $\eta_0$ and $\eta_1$ at fixed $t=t_0$,
\begin{align}
I(\eta_1,t_0)-I(\eta_0,t_0)
=\int_{0}^{1}\!\big\langle \nabla_{\bm \eta} 
I(\bm \eta(s),t_0),\,\dot {\bm \eta}(s)\big\rangle\,ds.
\end{align}
\end{itemize}
This recovers the path-integral representation \eqref{eq:path-int} 
(taking, e.g., a straight-line path in $\bm \eta$).

\end{remark}

\section{Mutual Information Maximization} 
\label{sec:maximization}

Having established a tractable formula for the information gradient 
in Section~\ref{sec:general}, 
we now detail how it can be employed to optimize the channel parameters $\bm \eta$ 
for maximizing mutual information.

\subsection{Gradient-based Optimization Framework} 
\label{subsec:gradient_ascent}

We aim to solve the optimization problem posed in \eqref{eq:optimization_problem}. 
Although various types of constraints on the front-end parameters $\bm\eta$ can be considered, 
we focus on the following regularized optimization problem for simplicity and clarity\footnote{We will discuss a projected gradient ascent method in Subsection \ref{sec:mi_design} for a constrained optimization problem.}.
To incorporate potential constraints or preferences on the parameters $\bm \eta$, 
we consider the regularized maximization problem: 
\begin{align} 
\text{maximize }_{\bm\eta \in \R^d}\ \mathcal{J}(\bm\eta) \equiv I(X;Y_t) - \lambda C(\bm\eta), 
\label{eq:max_objective} 
\end{align} 
where $C(\cdot)$ is a differentiable regularizer (or penalty function) 
encouraging desirable properties for $\bm \eta$ 
(e.g., sparsity, power constraints encoded softly), 
and $\lambda \geq 0$ controls the regularization strength. 

Using the information gradient formula \eqref{eq:gradient_formula} 
derived in Proposition~\ref{prop:info_gradient}, 
the gradient of the objective function $\mathcal{J}(\bm \eta)$ is given by
\begin{align} 
\nabla_{\bm\eta} \mathcal{J}(\bm\eta) &= \nabla_{\bm\eta} I(X;Y_t) 
- \lambda \nabla_{\bm\eta} C(\bm\eta) \\
&= -\E\!\left[ D f_{\bm\eta}(X)^{\!\top} s_{Y_t}\!\big(\Yt\big)\right]
- \lambda \nabla_{\bm\eta} C(\bm\eta).
\label{eq:objective_gradient} 
\end{align} 
We can then perform stochastic gradient ascent. 
Using a mini-batch estimate $\widehat{\E}[\cdot]$ for the expectation term, 
computed efficiently via the VJP identity 
as shown in \eqref{eq:batch_vjp} with a learned score model $s_{\bm \theta}(\cdot)$.
The update rule of the parameters $\bm \eta$ is given by
\begin{align}
    \label{eq:param_update}
    \bm\eta \leftarrow \bm\eta - \alpha
    \Big\{ \widehat{\E}\!\left[ D f_{\bm\eta}(X)^{\!\top} s_{\bm \theta}(\Yt) \right]
           + \lambda \nabla_{\bm\eta} C(\bm\eta)\Big\},
\end{align}
where $\widehat{\E}[\cdot]$ denotes a minibatch estimate with the VJP implementation,
and $\alpha>0$ is the learning rate.

A crucial point arises here: the true marginal score $s_{Y_t}(\cdot)$ 
implicitly depends on the parameters $\bm \eta$ because the marginal distribution $p_{Y_t}(\cdot)$ 
changes as $\bm \eta$ is updated. Therefore, the score model $s_{\bm \theta}(\cdot)$ 
used in the gradient estimate \eqref{eq:param_update} should ideally reflect 
the score corresponding to the current parameter value $\bm \eta$. 
This necessitates periodically updating or re-training the score model $s_{\bm \theta}(\cdot)$ 
as $\bm \eta$ evolves. The overall procedure thus becomes an alternating optimization, 
detailed next.

\subsection{Alternating Optimization}
\label{subsec:alternating_optimization}

The details of the alternating optimization are given in Algorithm~\ref{alg:alternating_opt}.
The optimization procedure alternates between updating the score model $s_{\bm \theta}$ 
and the front-end parameters $\bm \eta$ to maximize the regularized objective 
$\mathcal{J}(\bm\eta) = I(X;Y_t) - \lambda C(\bm\eta)$ defined in \eqref{eq:max_objective}. 

\begin{algorithm}
\caption{Alternating Optimization for MI Maximization}
\label{alg:alternating_opt}
\begin{algorithmic}[1]
\Require Initial parameters $\bm \eta^{(0)}$, score model parameters $\bm \theta^{(0)}$, learning rates $\alpha_\theta, \alpha_\eta$, number of outer iterations $K$, number of inner score-learning steps $S$.
\Ensure Optimized parameters $\bm \eta^{(K)}$.
\For{$k = 0$ to $K-1$}

\State {\textbf{Phase 1: $\bm \theta$-Update} (fix $\bm \eta^{(k)}$, update $\bm \theta$)}

\State Initialize or load score model parameters $\bm \theta \leftarrow \bm \theta^{(k)}$
    \For{$s = 1$ to $S$}
        \State Sample a mini-batch $\{\bm x_i\}_{i=1}^B$ and $\{\bm \varepsilon_i\}_{i=1}^B$,\\
        \hspace{1cm}$\bm x_i \sim p_X, \bm \varepsilon_i \sim \mathcal{N}(0, \bm I_m)$.
        \State Compute $\bm w_i = f_{\bm \eta^{(k)}}(\bm x_i)$ for $i=1,\dots,B$.
        \State Compute the DSM loss using the current $\bm \theta$:
            \[ \mathcal{L}_{\text{dsm}}(\bm \theta) \equiv \frac{1}{B} \sum_{i=1}^B \left\|s_{\bm \theta}(\bm w_i + \sqrt{t}\,\bm \varepsilon_i) + \frac{\bm \varepsilon_i}{\sqrt{t}}\right\|^2. \]
        \State Update score model parameters:
            \[ \bm \theta \leftarrow \bm \theta - \alpha_\theta \nabla_{\bm \theta} \mathcal{L}_{\text{dsm}}(\bm \theta). \]
    \EndFor
    \State Store updated score parameters $\bm \theta^{(k+1)} \leftarrow \bm \theta$.

    \State {\textbf{Phase 2: $\bm \eta$-Update} (fix $\bm \theta^{(k+1)}$, update $\bm \eta^{(k)}$)}
    \State Sample a mini-batch $\{\bm x_i\}_{i=1}^B$ and $\{\bm z_i\}_{i=1}^B$,\\
    \hspace{0.5cm}$\bm x_i \sim p_X, \bm z_i \sim \mathcal{N}(0, t \bm I_m)$.
    \State Compute $\bm y_i = f_{\bm \eta^{(k)}}(\bm x_i) + \bm z_i$ for $i=1,\dots,B$.
    \State Evaluate the fixed score function: $\bm s_i = s_{\bm \theta^{(k+1)}}(\bm y_i)$.
    \State Define the VJP loss using the current $\bm \eta^{(k)}$:
        \[ \mathcal{L}_{\text{vjp}}(\bm \eta^{(k)}) = \frac{1}{B}\sum_{i=1}^B \langle f_{\bm \eta^{(k)}}(\bm{x}_i), \text{stop}(\bm s_i) \rangle. \]
    \State Estimate the gradient of the objective: 
        \[ \widehat{\nabla_{\bm\eta} \mathcal{J}(\bm \eta^{(k)})} = - \nabla_{\bm \eta} \mathcal{L}_{\text{vjp}}(\bm \eta^{(k)}) - \lambda \nabla_{\bm \eta} C(\bm \eta^{(k)}). \]
    \State Update front-end parameters using gradient ascent:
        \[ \bm \eta^{(k+1)} \leftarrow \bm \eta^{(k)} + \alpha_\eta \widehat{\nabla_{\bm\eta} \mathcal{J}(\bm \eta^{(k)}).} \]
\EndFor
\end{algorithmic}
\end{algorithm}

In general, the objective function $I(X;Y_t)$ is non-concave with respect to 
$\bm \eta$. Therefore, this alternating procedure typically converges to a stationary point 
of the regularized objective $\mathcal{J}(\bm\eta)$. Under specific conditions, 
such as problem-specific convexity (which is rare but can occur in simplified linear-Gaussian scenarios) 
and appropriate step-size control, convergence to a global optimum might be possible, 
but local optimality is more generally expected.

\subsection{Computational Advantages}
\label{subsec:computational_advantages}
The VJP-based implementation of the information gradient offers significant 
computational advantages, particularly for models with high-dimensional parameter 
spaces or output spaces. 
A key benefit is that it {\em avoids the explicit computation and storage of 
the potentially massive $m \times d$ Jacobian matrix} 
$D f_{\bm \eta}(\bm x)$. 
Instead, the required vector-Jacobian product is computed directly, 
typically requiring only about the same computational effort as a single 
reverse-mode automatic differentiation pass 
through the function $f_{\bm \eta}(\bm x)$, 
a substantial saving compared to the $O(md)$ cost of forming the full Jacobian. 
Furthermore, the core operation within the VJP loss, 
the inner product $\langle f_{\bm \eta}(\bm x_i), s_{\bm \theta}(\bm y_i) \rangle$, 
is {highly amenable to vectorization}, allowing efficient computation across 
large mini-batches on parallel hardware such as GPUs.

Despite these advantages making the gradient \emph{estimation} step efficient, 
it is important to consider the overall computational cost of the alternating optimization 
procedure in Algorithm~\ref{alg:alternating_opt}. 
The score learning phase, which requires retraining or fine-tuning the score model 
$s_{\bm \theta}$ periodically as $\bm \eta$ updates, 
can still be computationally intensive. In some practical scenarios, 
this score learning step may become the primary bottleneck 
for the optimization process.

One remedy for avoiding the expensive retraining in Phase~1 
is to employ a {\em conditional score network}. Such a network is designed 
to approximate the score function conditioned on the parameters, 
taking both the noisy output $\bm y$ and the parameters $\bm \eta$ as input, 
i.e., $s_{\bm \theta}(\bm y, \bm \eta)$. 
This network can be trained offline (or initially) 
to generalize across a range of relevant parameter values by minimizing 
an objective:
\begin{align}
\E_{\bm x, \bm \varepsilon, \bm \eta}\left[\left\|s_{\bm \theta}
(f_{\bm \eta}(\bm x)+\sqrt{t} \bm \varepsilon, \bm \eta)
+ \frac{\bm \varepsilon}{\sqrt{t}}\right\|^2\right],
\label{eq:conditional_dsm}
\end{align}
where $\bm \varepsilon \sim \mathcal{N}(\bm{0}, \bm I_m)$ and the 
expectation includes averaging over parameters $\bm \eta$ 
drawn from an appropriate prior or training distribution. 
If this conditional network $s_{\bm \theta}(\bm y, \bm \eta)$ is sufficiently trained, 
it can serve as a proxy for the true score function $s_{Y_t}(\bm y)$ corresponding 
to any given $\bm \eta$ within the trained range. 
Consequently, the explicit {Phase~1 (Score Learning) of 
the alternating optimization procedure can potentially be skipped} 
at each iteration. Instead, one can directly use the pre-trained conditional score network 
$s_{\bm \theta}(\bm y, \bm \eta^{(k)})$, 
evaluated with the current parameters $\bm \eta^{(k)}$, 
within Phase~2 to estimate the gradient $\widehat{\nabla_{\bm\eta} \mathcal{J}(\bm \eta^{(k)})}$ 
and update $\bm \eta$.

\section{Task-Oriented Extension}
\label{sec:task}

In many practical scenarios, perfect reconstruction of the input signal $X$ at the receiver 
is unnecessary or even undesirable. Instead, the primary goal is often to preserve information 
relevant to a specific task variable, denoted by $T = g(X)$. 
Here, $g: \R^n \to \mathcal{T}$ represents a function that extracts task-relevant features 
from the input $X$, mapping it to a task space $\mathcal{T}$. 
Examples of such tasks abound in modern applications, 
including image classification where $T$ represents discrete class 
labels (e.g., $T \in \{1,\ldots,K\}$), anomaly detection where $T$ might be 
a binary indicator ($T \in \{0,1\}$), 
or parameter estimation (regression) where $T$ consists 
of continuous target values ($T \in \R^k$). 

It is therefore natural to shift the optimization objective 
from maximizing $I(X;Y_t)$ to maximizing the task-relevant 
information $I(T;Y_t)$. This formulation aligns directly with 
emerging concepts in semantic and goal-oriented communication \cite{Gunduz2023}. 
In this section, we extend our information gradient framework 
to address these task-oriented optimization problems.

\subsection{Task-Oriented SFB: Fisher Integral Representation of Task-Oriented Mutual Information}
\label{subsec:sfb-task-fisher-int}

To apply the principles underlying the SFB framework to the task-oriented 
objective $I(T;Y_t)$, we first require a method to connect this mutual information 
to score-based quantities. 
Analogous to how $I(X;Y_t)$ relates to the marginal Fisher information via the de Bruijn identity, 
the task-oriented MI involves both the marginal Fisher information $J(Y_t)$ 
and the Fisher information conditioned on the task variable, $\mathbb{E}_T[J(Y_t|T)]$. 
This subsection derives the resulting Fisher integral representation for $I(T;Y_t)$, 
establishing the necessary foundation for both estimation and optimization within the SFB context.

Let $T=g(X)$ (discrete or continuous) and 
$Y_t=f_{\bm\eta}(X)+Z_t$ with $Z_t\sim\mathcal N(\bm 0, t \bm I_m)$ for $t>0$.
Define the score functions:
\begin{align}
s_{Y_t}(\bm y) &\equiv \nabla_{\bm y}\log p_{Y_t}(\bm y), \\
s_{Y_t| T}(\bm y| \tau ) &\equiv \nabla_{\bm y}\log p_{Y_t| T}(\bm y |\tau),
\end{align}
where the variable $\tau$ represents a realization of the task variable $T$.
The corresponding Fisher informations are given by
\begin{align}
J(Y_t) &\equiv \E_{Y_t}\!\big[\|s_{Y_t}(Y_t)\|^2\big], \\
J(Y_t |T) &\equiv \E_{Y_t |T}\!\big[\|s_{Y_t | T}(Y_t| T)\|^2\,\big|\,T\big].
\end{align}
By the unconditional and conditional de~Bruijn identities:
\begin{align}
\frac{d}{dt} h(Y_t) &=\frac12 J(Y_t), \\
\frac{d}{dt} h(Y_t| T) &=\frac12\,\E_T\!\big[J(Y_t| T)\big],
\end{align}
we obtain the task-oriented derivative
\begin{align}
\frac{d}{dt}I(T;Y_t)
= \frac12\Big(J(Y_t)-\E_T\!\big[J(Y_t | T)\big]\Big).
\label{eq:sfb-task-diff}
\end{align}
Integrating over $[t^*,U]$ yields
\begin{align}
I(T;Y_{t^*})
= \frac12 \int_{t^*}^{U}\!\Big(J(Y_t) - \E_T[J(Y_t\mid T)]\Big)\,\mathrm dt
\;+\; I(T;Y_U).
\label{eq:sfb-task-int-U}
\end{align}
In the strong-noise limit $U\to\infty$ 
we can obtain the Fisher integral representation:
\begin{align}
I(T;Y_{t^*})
= \frac12 \int_{t^*}^{\infty}\!\Big(J(Y_t) - \E_T[J(Y_t\mid T)]\Big)\,\mathrm dt
\label{eq:sfb-task-int}
\end{align}
under the regularity condition that the expectation is finite.

If we have score models for $s_{Y_t}(\bm y)$ and $s_{Y_t|T}(\bm y|\tau)$, 
we can estimate the task-oriented mutual information by Monte Carlo estimators of 
the Fisher and conditional Fisher information
according to the standard SFB recipe.


\begin{remark}[Reduction to the standard SFB when $T=X$] \rm
    \label{rem:sfb-reduction}
    Setting $T=X$ in \eqref{eq:sfb-task-diff} yields the classical SFB for $I(X;Y_t)$.
    Indeed, conditional PDF on $X=\bm x$ is given by
    \begin{align}
    p_{Y_t| X}(\bm y|\bm x) = \mathcal N\!\big(f_{\bm\eta}(\bm x),\, t \bm I_m\big),
    \end{align}
    so the conditional score is given by
    \begin{align}
    s_{Y_t\mid X}(\bm y|\bm x)=\nabla_{\bm y}\log p_{Y_t|X}(\bm y|\bm x)
    = -\frac{1}{t}\big(\bm y - f_{\bm\eta}(\bm x)\big),
    \end{align}
    and therefore
    \begin{align}
    J(Y_t|X)
    &=\E\!\big[\|s_{Y_t\mid X}(Y_t\mid X)\|^2\,\big|\,X\big] \\
    &=\frac{1}{t^2}\,\E\!\big[\|Y_t-f_{\bm\eta}(X)\|^2 \,\big|\, X\big] \\
    &=\frac{m}{t}.
    \end{align}
    Taking expectation over $X$ gives $\E[J(Y_t| X)]=m/t$, and
    the equation \eqref{eq:sfb-task-diff} reduces to
    \begin{align}
    \frac{d}{dt}\,I(X;Y_t)
    =\frac12\Big(J(Y_t)-\frac{m}{t}\Big).
    \end{align}
    Integrating (with the same regularity) yields the standard Fisher integral representation:
    \begin{align}
    I(X;Y_{t^*})
    =\frac12\int_{t^*}^{\infty}\!\Big(J(Y_t)-\frac{m}{t}\Big)\,dt,
    \end{align}
    where $m=\dim(Y)$.
\end{remark}
   
\begin{remark}[Stochastic tasks $T\!\sim\! p(T|X)$] \rm 
    \label{rem:stochastic_T}
    While this paper focuses on deterministic tasks $T=g(X)$, all SFB and gradient 
    formulas extend verbatim to the stochastic case $T\!\sim\!p_{T|X}$ provided that 
    (i) forward sampling from $p_{T|X}$ is efficient and 
    (ii) standard regularity conditions for exchanging differentiation and expectation hold. 
    In particular, the conditional de~Bruijn identity remains valid:
    \begin{align}
    \frac{d}{dt}I(T;Y_t) = \frac12\Big(J(Y_t)-\E_{T}\big[J(Y_t| T)\big]\Big),
    \end{align}
    hence the Fisher integral representation and its Monte Carlo estimation 
    carry over with expectations taken jointly over $(X,T)$.
\end{remark}

\subsection{Task-Oriented Information Gradient}
\label{subsec:task_gradient}

The information gradient corresponding to $I(T;Y_t)$ can be used 
for optimizing the parameter $\bm \eta$ to maximize 
the task-oriented mutual information.

\begin{proposition}[Task-oriented information gradient]
\label{prop:task_gradient}
Under standard regularity conditions, the gradient of 
the task-oriented mutual information is given by
\begin{align}
\nabla_{\bm \eta} I(T; \Yt) = \E_{X,Z_t}\left[D f_{\bm \eta}(X)^\top \left(\sYtT(\Yt|T) 
- \sYt(\Yt)\right)\right],
\label{eq:task_gradient}
\end{align}
where $\Yt = f_{\bm \eta}(X) + Z_t$.
\end{proposition}
Note that we assumed $T=g(X)$;
this means that the expectation over $T$ is implicit since $T$ is determined by $X$.

\begin{proof}
The mutual information can be decomposed as:
\begin{align}
I(T; \Yt) 
&= h(\Yt) - h(\Yt|T)  \\
&= -\E_{\Yt}[\log p_{\Yt,\bm \eta}(\Yt)] + \E_{T,\Yt}[\log p_{\Yt|T,\bm \eta}(\Yt|T)].
\end{align}
Differentiating with respect to $\bm \eta$ 
and using the reparameterization $\Yt = f_{\bm \eta}(X) + Z_t$, we have
\begin{align}
\nabla_{\bm \eta} I(T; \Yt)
&= -\nabla_{\bm \eta} \E_{X,Z_t}[\log p_{\Yt,\bm \eta}(f_{\bm \eta}(X) + Z_t)] \nonumber \\
&\quad + \nabla_{\bm \eta} \E_{X,Z_t}[\log p_{\Yt|T,\bm \eta}(f_{\bm \eta}(X) + Z_t | T)].
\end{align}
Applying the chain rule to each term and using the score property 
$\E[\nabla_{\bm \eta} \log p] = 0$ to eliminate the parameter-dependent terms 
(as in the proof of Proposition~\ref{prop:info_gradient}), 
we obtain
\begin{align} \nonumber
&\nabla_{\bm \eta} I(T; \Yt) \\
&= \E_{X,Z_t}\left[D f_{\bm \eta}(X)^\top \nabla_{\bm y} \log p_{\Yt|T,\bm \eta}(\bm y|T)\big|_{\bm y=f_{\bm \eta}(X)+Z_t}\right] \nonumber \\
&\quad - \E_{X,Z_t}\left[D f_{\bm \eta}(X)^\top \nabla_{\bm y} \log p_{\Yt,\bm \eta}(\bm y)\big|_{\bm y=f_{\bm \eta}(X)+Z_t}\right],
\end{align}
which yields \eqref{eq:task_gradient}. 
\end{proof}

It is important to note that the conditional score expectation vanishes when $T = X$, recovering~\eqref{eq:gradient_formula}.
Moving $\bm \eta$ in the direction of $D f_{\bm \eta}^\top (\sYtT - \sYt)$ 
increases $I(T; \Yt)$. This principle is the key to the {task-oriented optimization}.

\begin{remark}[Stochastic tasks $T\!\sim\! p(T|X)$] \rm 
In case of stochastic tasks $T\!\sim\! p(T|X)$, the information gradient becomes
\begin{align}
\nabla_{\bm\eta} I(T;Y_t)
=
\E_{X,Z_t,T}\!\left[
D f_{\bm\eta}(X)^{\!\top}
\big\{ s_{Y_t|T}(Y_t|T)-s_{Y_t}(Y_t) \big\}
\right],
\end{align}
i.e., the random variable $T$ is included in the expectation.
Practically, one trains the unconditional score $s_{Y_t}$ and the conditional score $s_{Y_t|T}$ using paired samples $(X,T)$ obtained by forward sampling $T\!\sim\!p(T|X)$, 
then uses the same VJP estimator as in the deterministic case.
\end{remark}

\subsection{Implementation via Alternating Optimization}
\label{subsec:task_implementation}

Optimizing the task-oriented objective using the gradient \eqref{eq:task_gradient} 
follows the same alternating optimization structure outlined 
in Algorithm~\ref{alg:alternating_opt}. 
The key difference lies in the score learning phase (Phase 1). 
Instead of learning only the marginal score $s_{Y_t}$, 
we now need approximations for both the marginal score
 $s_{Y_t}(\bm y) \approx s_{\bm \phi}(\bm y)$ and 
 the conditional score $s_{Y_t|T}(\bm y|\tau) \approx s_{\bm \psi}(\bm y, \tau)$.

These are learned using appropriate DSM objectives. 
The marginal score $s_{\bm \phi}(\bm y)$ is trained using the standard DSM loss:
\begin{align}
\E_{\bm{x}, \bm{\varepsilon}}\left[\left\|s_{\bm \phi}(\bm{w} + \sqrt{t}\,\bm{\varepsilon}) + \frac{\bm{\varepsilon}}{\sqrt{t}}\right\|^2\right], \quad \bm w = f_{\bm \eta}(\bm x).
\label{eq:unconditional_dsm}
\end{align}
The conditional score $s_{\bm \psi}(\bm y, \tau)$ is trained using a conditional 
DSM loss, where only $\bm y$ is perturbed while 
the conditioning variable $\tau=g(\bm x)$ remains clean:
\begin{align}
\E_{\bm{x}, \bm{\varepsilon}}\left[\left\|s_{\bm \psi}(\bm{w} + \sqrt{t}\,\bm{\varepsilon}, \tau) + \frac{\bm{\varepsilon}}{\sqrt{t}}\right\|^2\right], \quad \bm w = f_{\bm \eta}(\bm x), \tau = g(\bm x).
\label{eq:conditional_dsm_task}
\end{align}
In Phase 2 (Parameter Update), these two learned score functions, 
$s_{\bm \phi}$ and $s_{\bm \psi}$, are then used within the VJP framework 
to estimate the task-oriented gradient \eqref{eq:task_gradient} for updating $\bm \eta$.

Given learned score functions $s_{\bm \phi}$ and $s_{\bm \psi}$, we update $\bm \eta$ using the VJP formulation. For a mini-batch $\{(\bm x_i, \bm z_i)\}_{i=1}^B$ with $\tau_i = g(\bm x_i)$, compute $\bm w_i = f_{\bm \eta}(\bm x_i)$ and $\bm y_i = \bm w_i + \bm z_i$. Define the VJP loss as in \eqref{eq:vjp_loss}:
\begin{align}
\mathcal{L}_{\text{vjp}}^{\text{task}} \equiv \frac{1}{B}\sum_{i=1}^B \left\langle f_{\bm \eta}(\bm{x}_i), \text{stop}\left(s_{\bm \psi}(\bm{y}_i, \tau_i) - s_{\bm \phi}(\bm{y}_i)\right) \right\rangle.
\label{eq:task_vjp_loss}
\end{align}
The gradient estimator is thus given by
\begin{align}
\widehat{\nabla_{\bm \eta} I(T; \Yt)} &\equiv \nabla_{\bm \eta} \mathcal{L}_{\text{vjp}}^{\text{task}}.
\end{align}

\subsection{Incorporating Utility Functions}
\label{subsubsec:utility}

In certain applications, maximizing the raw mutual information $I(T;Y_t)$ 
might not directly align with the ultimate system goal. 
For instance, the perceived benefit or `utility' might exhibit diminishing 
returns as MI increases, 
or perhaps the goal is related to a downstream task performance metric 
that is a nonlinear function of the MI. 
To accommodate such scenarios, we can introduce a differentiable utility 
function $U: \R \to \R$, where $U(I)$ represents 
the utility derived from achieving a mutual information level $I$. 
The objective then becomes maximizing this utility, $U(I(T;Y_t))$.

Thanks to the chain rule, computing the gradient for this new objective 
is straightforward. The gradient with respect to the parameters $\bm \eta$ 
is simply:
\begin{align}
\nabla_{\bm \eta} U(I(T;Y_t)) = U'(I(T;Y_t)) \cdot \nabla_{\bm \eta} I(T;Y_t),
\label{eq:utility_gradient}
\end{align}
where $U'(I)$ is the derivative of the utility function evaluated 
at the current MI level $I=I(T;Y_t)$. 
In practice, we can estimate this gradient by multiplying our 
previously derived estimator 
for $\nabla_{\bm \eta} I(T;Y_t)$ (e.g., $\nabla_{\bm \eta} 
\mathcal{L}_{\text{vjp}}^{\text{task}}$ 
from \eqref{eq:task_vjp_loss}) by 
the scalar value $U'(I)$. 
Note that estimating the current MI value $I(T;Y_t)$ 
might be needed to evaluate $U'(I)$, 
potentially using methods like the SFB integral 
representation \eqref{eq:sfb-task-int} or path integral \eqref{eq:path-int} 
if $U'$ is not constant. 
This allows the optimization framework to be adapted to 
a wider range of application-specific objectives beyond direct MI maximization.

\section{Information Bottleneck Extension}
\label{sec:ib}

In practical communication systems, bandwidth is a critical and often scarce resource. 
While the task-oriented objective of maximizing $I(T; Y_t)$ ensures 
high performance related to the task variable $T=g(X)$, 
it might inadvertently lead to representations $Y_t$ 
that retain excessive information about the original input $X$, 
potentially consuming unnecessary bandwidth. 
The information bottleneck (IB) principle \cite{Tishby1999} 
offers an elegant framework to navigate this trade-off 
by explicitly balancing two competing goals: maximizing relevance, 
measured by $I(T; Y_t)$, while simultaneously promoting compression 
by minimizing $I(X; Y_t)$. This encourages the system to learn 
representations that are maximally informative about the task $T$ 
while being minimally informative about the rest of the input $X$. 
In this section, we extend the proposed information gradient framework 
to optimize systems according to the IB criterion.

The IB objective \cite{Tishby1999} is defined as:
\begin{equation}
\mathcal{L}_{\text{IB}}(\bm \eta) = I(T; Y_t) - \beta I(X; Y_t),
\label{eq:ib_objective}
\end{equation}
where $\beta \geq 0$ is a Lagrange multiplier controlling 
the trade-off between compression and relevance. 

Combining the information gradient formulas from 
Proposition~\ref{prop:info_gradient} and 
Proposition~\ref{prop:task_gradient}, 
we can derive the gradient of the IB objective.
\begin{proposition}[IB gradient]
\label{prop:ib_gradient}
Under standard regularity conditions, the gradient 
of the IB objective~\eqref{eq:ib_objective} with respect to channel parameters $\bm \eta$ is given by
\begin{equation}
\begin{split}
\nabla_{\bm \eta} \mathcal{L}_{\text{IB}}(\bm \eta) = \mathbb{E}_{X,Z_t}\Big[ Df_\eta(X)^\top \big( &s_{Y_t|T}(Y_t|T) \\
&+(\beta-1) s_{Y_t}(Y_t) \big) \Big],
\end{split}
\label{eq:ib_gradient}
\end{equation}
where $Y_t = f_{\bm \eta}(X) + Z_t$, $s_{Y_t}(\bm y) = \nabla_{\bm y} \log p_{Y_t}(\bm y)$ is the unconditional score function, and $s_{Y_t|T}(\bm y|\tau) = \nabla_{\bm y} \log p_{Y_t|T}(\bm y|\tau)$ is the conditional score function.
\end{proposition}

\begin{proof}
From the definition~\eqref{eq:ib_objective}, we have:
\begin{equation}
\nabla_\eta \mathcal{L}_{\text{IB}} = \nabla_\eta I(T; Y_t) - \beta \nabla_\eta I(X; Y_t).
\label{eq:ib_gradient_decomp}
\end{equation}

Applying Proposition~\ref{prop:task_gradient} (task-oriented gradient), we have
\begin{equation}
\nabla_\eta I(T; Y_t) = \mathbb{E}_{X,Z_t}\left[Df_\eta(X)^\top \left(s_{Y_t|T}(Y_t|T) - s_{Y_t}(Y_t)\right)\right].
\label{eq:task_gradient_ib}
\end{equation}
In a similar manner, applying Proposition~\ref{prop:info_gradient}, we obtain
\begin{equation}
\nabla_\eta I(X; Y_t) = -\mathbb{E}_{X,Z_t}\left[Df_\eta(X)^\top s_{Y_t}(Y_t)\right].
\label{eq:general_gradient_ib}
\end{equation}
Substituting~\eqref{eq:task_gradient_ib} and~\eqref{eq:general_gradient_ib} into~\eqref{eq:ib_gradient_decomp}, we obtain the claimed result~\eqref{eq:ib_gradient}.
\begin{equation}
\begin{split}
\nabla_\eta \mathcal{L}_{\text{IB}} &= \mathbb{E}_{X,Z_t}\left[Df_\eta(X)^\top \left(s_{Y_t|T} - s_{Y_t}\right)\right] \\
&\quad + \beta \mathbb{E}_{X,Z_t}\left[Df_\eta(X)^\top s_{Y_t}\right] \\
&= \mathbb{E}_{X,Z_t}\left[Df_\eta(X)^\top \left(s_{Y_t|T} + (\beta-1) s_{Y_t}\right)\right].
\end{split}
\end{equation}
\end{proof}

\begin{remark}[Reduction to Previous Results] \rm 
Setting $\beta = 0$ in~\eqref{eq:ib_gradient} recovers 
Proposition~\ref{prop:task_gradient} (task-oriented gradient). 
The IB framework thus provides a unified view encompassing both 
general and task-oriented optimization.
\end{remark}

The IB information gradient is evaluated via the VJP identity 
and its implementation follows Subsection~\ref{subsec:task_implementation}
(two-model score learning, stop-gradient).

\section{Numerical Experiments}
\label{sec:experiments}

\subsection{Scalar Linear-Gaussian Channel}

\subsubsection{Analytical Validity Check}
\label{subsec:linear_gaussian}
To assess the correctness of the proposed information gradient formula, 
we consider the analytically tractable setting with a scalar linear-Gaussian channel 
with dimension $n=m=1$:
\begin{align}   
X \sim \mathcal N(0,\sigma_X^2), \ 
Y_t = \alpha X+Z_t, \ 
Z_t \sim \mathcal N(0,t),
\end{align}
with $\alpha\in\mathbb R$ and noise variance $t>0$, i.e.,
the parameter $\alpha$ is treated as the front-end parameter.
In this case, the mutual information is given by
\begin{equation}
\label{eq:MI-LG}
I(X;Y_t)=\frac12\log\!\Big(1+\frac{\alpha^2\sigma_X^2}{t}\Big).
\end{equation}
Differentiating \eqref{eq:MI-LG} with respect to $\alpha$ yields the analytic gradient
\begin{equation}
\label{eq:dI-dalpha-analytic}
\frac{\partial I}{\partial \alpha}
=\frac{\alpha\,\sigma_X^2}{\,t+\alpha^2\sigma_X^2\,}.
\end{equation}

We now specialize our informaiton gradient formula (\ref{eq:gradient_formula})
to the scalar parameterization $f_\eta(x)=\alpha x$ (so that $D f_\eta(X)=X$ and $\eta=\alpha$).
Since $Y_t$ is Gaussian with variance
$v=\mathrm{Var}(Y_t)=\alpha^2\sigma_X^2+t$, its marginal score is 
$s_{Y_t}(y)=\partial_y\log p_{Y_t}(y)=-y/v$. Therefore, we have
\begin{align*}
    \nabla_\eta I(X;Y_t) = \frac{\partial I}{\partial \alpha}
&= -\mathbb E\!\Big[D f_\eta(X)^\top s_{Y_t}(Y_t)\Big] \\
&=-\mathbb E\!\left[X\Big(-\frac{Y_t}{v}\Big)\right] \\
&=\frac{\mathbb E[XY_t]}{v} =\frac{\alpha\,\sigma_X^2}{t+\alpha^2\sigma_X^2},
\end{align*}
which coincides exactly with the analytic result \eqref{eq:dI-dalpha-analytic}. 
This equality verifies that the proposed information gradient expression 
recovers the correct answer in the scalar linaer-Gaussian model.

\subsubsection{Experimental Setup}

To validate the proposed information gradient formula 
and the path-integral route to mutual information, we conduct numerical experiments 
on the scalar linear-Gaussian channel, where analytical solutions are available for comparison.

We consider the scalar linear-Gaussian channel model
described in Subsection~\ref{subsec:linear_gaussian}.
For the numerical experiments,
we set $\sigma_X = 1.0$, $t = 0.5$, and vary $\alpha$
over the interval $[0, 3]$ with 61 equally-spaced points. 
We estimate the information
gradient using the VJP-based score estimator with 
$N = 200{,}000$ Monte Carlo samples.


Applying Proposition \ref{prop:info_gradient} with 
$f_\eta(x) = \alpha x$ 
and $Df_\eta(x) = x$, the gradient
estimator becomes
\begin{equation}
\widehat{\nabla_\alpha I} = -\frac{1}{N} \sum_{i=1}^N X_i 
\cdot s_{Y_t}(Y_{t,i}) =
\frac{1}{N} \sum_{i=1}^N \frac{X_i Y_{t,i}}{\alpha^2 \sigma_X^2 + t},
\end{equation}
where $\{X_i, Y_{t,i}\}_{i=1}^N$ are i.i.d. samples from the joint distribution. Note that, in this case, we use the exact score function $s_{Y_t}(y)$ instead of the estimated score function $s_{\bm \theta}(y)$.

\subsubsection{Path-Integral Reconstruction of Mutual Information}

Given the gradient estimates 
$\{\widehat{\nabla_\alpha I}(\alpha_k)\}_{k=0}^{60}$ at
discrete values $\alpha_0 = 0, \alpha_1, \ldots, \alpha_{60} = 3$, we reconstruct the
mutual information via numerical integration using the trapezoidal rule:
\begin{equation}
\widehat{I}(\alpha_k) = \widehat{I}(\alpha_{k-1}) + \frac{1}{2} \left(
\widehat{\nabla_\alpha I}(\alpha_{k-1}) + \widehat{\nabla_\alpha I}(\alpha_k) \right)
(\alpha_k - \alpha_{k-1}),
\end{equation}
with the initial condition $\widehat{I}(\alpha_0) = 0$ (since $I(X; Y_t) = 0$ when
$\alpha = 0$).

\subsubsection{Results}

Figure~\ref{fig:linear_gaussian_grad} presents the information gradient $\partial I /
\partial \alpha$ versus $\alpha$. The proposed VJP/score estimator (orange
dashed line) closely matches the analytical gradient (blue solid line) across the
entire range of $\alpha$. The gradient exhibits a peak at approximately $\alpha \approx
0.7$, indicating the parameter value at which the rate of information gain is
maximized. As $\alpha$ increases beyond this point, the gradient decreases
monotonically, reflecting diminishing returns in mutual information per unit increase
in channel gain.

Figure~\ref{fig:linear_gaussian_I} shows the mutual information $I(X; Y_t)$ as a
function of $\alpha$. The blue solid line represents the analytical solution, while the
orange dashed line shows the path-integral estimate obtained from the proposed
gradient formula. The two curves are visually indistinguishable, demonstrating that the
numerical integration accurately recovers the mutual information from gradient
estimates.

These results provide a validity check in a setting where ground-truth solutions are
known, confirming that our framework reproduces the correct gradients and mutual
information values. 


\begin{figure}[t]
\centering
\includegraphics[width=0.48\textwidth]{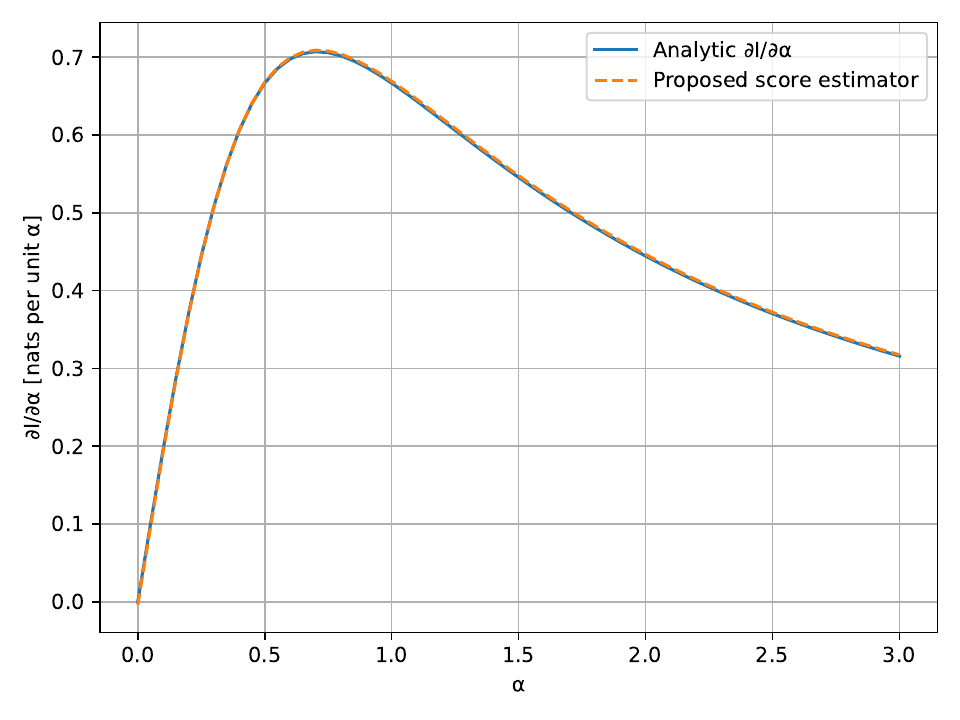}
\caption{Gradient of mutual information $\partial I / \partial \alpha$ versus $\alpha$.
The proposed VJP/score estimator (orange dashed) and the analytical
gradient (blue solid) are presented in the figure.}
\label{fig:linear_gaussian_grad}
\end{figure}

\begin{figure}[t]
    \centering
    \includegraphics[width=0.48\textwidth]{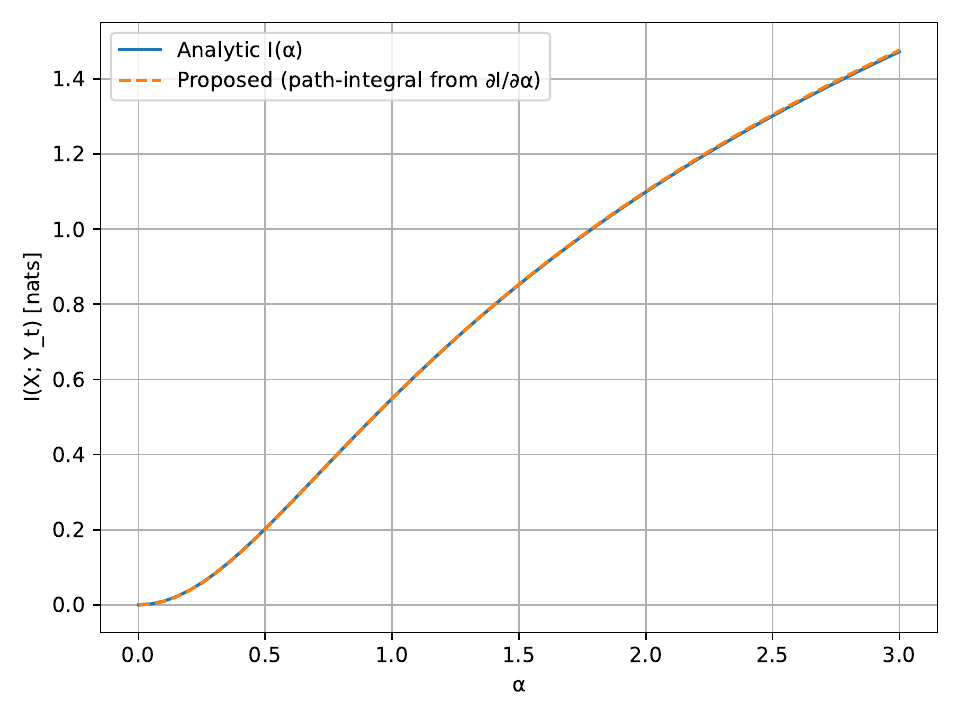}
    \caption{Mutual information $I(X; Y_t)$ versus channel gain $\alpha$ for the
    one-dimensional linear-Gaussian channel. The analytical solution (blue solid) and the
    path-integral estimate from gradient integration (orange dashed).}
    \label{fig:linear_gaussian_I}
\end{figure}

\subsection{Linear Vector Channel with Gaussian Input}
\label{subsec:lin_gauss_alphaA}

We consider the linear vector channel with Gaussian input:
$
Y = \alpha \bm A X + Z,
$
where $X\in\mathbb{R}^{n}$ is the input, 
$\bm A\in\mathbb{R}^{m\times n}$ is a fixed matrix, 
$\alpha\in\mathbb{R}$ is a scalar front-end parameter, 
and $Z\sim\mathcal{N}(0,t \bm I_m)$ with $t>0$ is additive white Gaussian noise independent of $X$. Throughout this subsection we focus on the Gaussian input case
$
X \sim \mathcal{N}(0,\bm \Sigma_x), \ \bm \Sigma_x=\sigma_x^2 \bm I_n \;\;(\sigma_x>0),
$
so that all relevant quantities admit closed form expressions.

\subsubsection{Preparation}
Under the above assumptions, we have
\begin{align}
\bm \Sigma_Y(\alpha) &= \mathrm{Cov}(Y) \\
 &= \alpha^2 \bm A \bm \Sigma_x \bm A^\top + t \bm I_m \\
 &= \alpha^2 \sigma_x^2 \bm A \bm A^\top + t \bm I_m.
\end{align}
We thus have the marginal score function
$
s_Y(\bm y) = \nabla_{\bm y} \log p_Y(\bm y) = -\bm \Sigma_Y(\alpha)^{-1} \bm y.
$

Since $(X,Y)$ are jointly Gaussian, the mutual information admits the standard log-det expression:
\begin{align}
\label{eq:I-alpha-logdet}
I(X;Y|\alpha) 
&= \frac{1}{2}\log\det\!\Big(\bm I_m + \frac{\alpha^2 \sigma_x^2}{t}\, \bm A \bm A^\top\Big).
\end{align}
Let $s_1,\dots,s_r$ be the nonzero singular values of $\bm A$ ($r=\mathrm{rank}(\bm A)$). Then \eqref{eq:I-alpha-logdet} can be written as the scalar sum
\begin{equation}
\label{eq:I-alpha-sum}
I(X;Y|\alpha) = \frac{1}{2}\sum_{i=1}^{r} \log\!\Big(1+\frac{\alpha^2 \sigma_x^2}{t}\, s_i^2\Big).
\end{equation}
Differentiating \eqref{eq:I-alpha-logdet} with respect to $\alpha$ and using the identity
\begin{equation}
\frac{d}{d\alpha}\log\det \bm M(\alpha)=\mathrm{tr}\!\big(\bm M^{-1}(\alpha) \bm M'(\alpha)\big)
\end{equation}
yields the gradient of the mutual information with respect to $\alpha$:
\begin{equation}
\label{eq:dIdalpha-trace}
\frac{\partial I}{\partial \alpha} 
= \frac{\alpha}{t}\,\mathrm{tr}\!\Big( \bm A \bm \Sigma_x \bm A^\top \,\big(\bm I_m + \tfrac{\alpha^2}{t} \bm A \bm \Sigma_x \bm A^\top\big)^{-1} \Big).
\end{equation}
Equivalently, using the singular values, we have
\begin{equation}
\label{eq:dIdalpha-sum}
\frac{\partial I}{\partial \alpha}
= \sum_{i=1}^{r}
\frac{\alpha\, \sigma_x^2 s_i^2}{\,t + \alpha^2 \sigma_x^2 s_i^2\,}.
\end{equation}

Let $f_\alpha(\bm x)=\alpha \bm A \bm x$. The VJP form of the information gradient is given by
\begin{equation}
\label{eq:vjp-estimator}
\frac{\partial I}{\partial \alpha}
= -\mathbb{E}\!\left[\left\langle \frac{\partial f_\alpha(X)}
{\partial \alpha}, s_Y(Y)\right\rangle \right]
= -\mathbb{E}\!\left[\big\langle \bm A X, s_Y(Y)\big\rangle \right],
\end{equation}
with $Y=\alpha \bm A X + Z$ and $s_Y(\bm y)=-\bm \Sigma_Y(\alpha)^{-1} \bm y$. 
In the Gaussian input case, plugging the closed-form $s_Y$ into 
\eqref{eq:vjp-estimator} and taking expectations recovers the analytic gradient 
\eqref{eq:dIdalpha-trace}–\eqref{eq:dIdalpha-sum}.
Practically, \eqref{eq:vjp-estimator} provides 
a Monte Carlo estimator of $\partial I/\partial\alpha$; 
integrating that estimator along $\alpha$ (with $I(\alpha{=}0)=0$) 
numerically reconstructs $I(\alpha)$.

\subsubsection{Experimental Setup}
\label{subsec:exp_lin_gauss}

We consider $X\in\mathbb{R}^n$, $Y\in\mathbb{R}^m$ with $n=m=8$, where $X\sim\mathcal{N}(0,\bm \Sigma_x)$, 
$\bm \Sigma_x=\sigma_x^2 \bm I_n$ with $\sigma_x^2=1$, and $Z\sim\mathcal{N}(0,t \bm I_m)$ with $t=0.5$. 
The front-end function  is $f_\alpha(\bm x)=\alpha \bm A \bm x$. 
The matrix $\bm A$ is generated as $\bm A = \bm U\operatorname{diag}(s)\bm V^\top$ 
using QR-based orthogonal $\bm U,\bm V$ 
and a mildly ill-conditioned geometric spectrum $s$ 
(largest-to-smallest ratio $\approx 10$-$15$). 
The analytic gradient used as ground truth is (\ref{eq:dIdalpha-sum}).
We used the  true score function $s_Y(\bm y)=- \bm \Sigma_Y(\alpha)^{-1} \bm y$ 
with $\bm \Sigma_Y(\alpha)=\alpha^2\sigma_x^2 \bm A \bm A^\top+t \bm I_m$, and estimate
\begin{align}
\frac{\partial I}{\partial\alpha}\approx -
\frac{1}{N}\sum_{k=1}^{N}\langle \bm A X^{(k)}, s_Y(Y^{(k)})\rangle,\quad N=10^5
\end{align}
in the Monte Carlo VJP estimator with the true score function.

For the learned-score VJP (DSM), we used \emph{unconditional} score models $s_\theta(\cdot)$ 
that do not take $\alpha$ or $\sigma$ as inputs; 
instead, an independent model is trained for each $\alpha$ 
(per-$\alpha$ training). 
We adopt DSM method at a \emph{fixed} noise level: 
$\sigma_{\rm train}=\sigma_{\rm eval}=0.1\sqrt t$. 
Each per-$\alpha$ model is a two-hidden-layer multi-layer perceptron (MLP) 
(width $256$, SiLU activations) trained 
for $1000$ epochs with batch size $4096$ 
using AdamW and gradient clipping. 
At test time, we apply a {\em Stein calibration} 
to correct the global scale of the learned score: 
we rescale $s_\theta$ by a scalar $c$ chosen to satisfy 
$\mathbb{E}\!\left[Y^\top \big(c\,s_\theta(Y)\big)\right]\approx 
-m$, which follows from the Gaussian Stein identity 
$\mathbb{E}\!\left[Y^\top s_Y(Y)\right]=-m$. 
Concretely, on the evaluation mini-batch we set
\begin{align}   
c = -\frac{m}{\frac{1}{N}\sum_{k=1}^{N}
Y^{(k)\top}s_\theta\!\big(Y^{(k)}\big)},
\end{align}
and use $c\cdot s_\theta$ in the VJP estimator.

\subsubsection{Results}

\begin{figure}[t]
  \centering
  \includegraphics[width=0.98\linewidth]{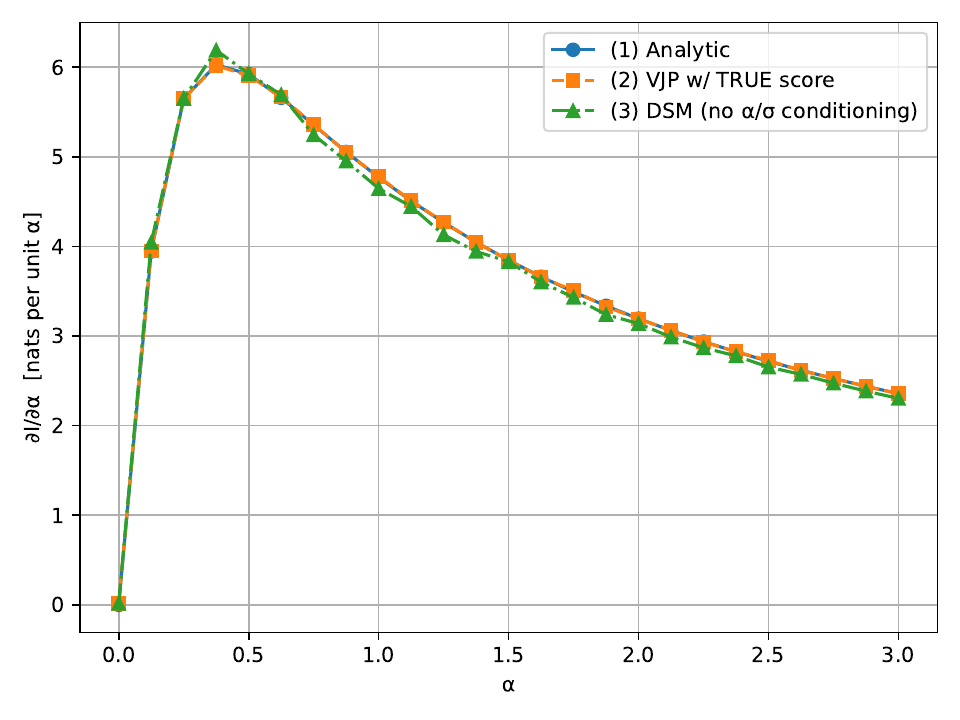}
  \vspace{-0.6em}
  \caption{Comparisons of the information gradient $\partial I/\partial\alpha$ versus $\alpha$ 
  for the linear vector channel with Gaussian input. The analytic gradient (blue solid), the true-score VJP (orange dashed), and the learned-score VJP (green dash-dotted) are compared. }
  \label{fig:grad_curve_peralpha}
\end{figure}

Figure \ref{fig:grad_curve_peralpha} shows the comparisons of 
the information gradient $\partial I/\partial\alpha$ versus $\alpha$ 
for the linear vector channel with Gaussian input. 
The analytic gradient (blue solid), the true-score VJP (orange dashed), 
and the learned-score VJP (green dash-dotted) are compared.
The VJP estimator with the closed form score is indistinguishable 
from the analytic curve up to Monte Carlo fluctuations, 
empirically confirming the VJP identity in this setting.
In our setting ($n=m=8$; $1000$ epochs; hidden $256$), 
the DSM curve visually overlaps with the analytic/true-score curves 
except for small residual noise. 
This can be considered as an strong empirical evidence for 
the effectiveness of the score-based estimation with VJP identity.
This result empirically validates the proposed VJP estimator 
and supports the practical use of learned scores for MI maximization.


\subsection{Gradient-based Optimization via Information Gradients}
\label{sec:mi_design}

Mutual information (MI) provides an intrinsic, task-agnostic measure of the informativeness 
of a sensing/communication front end. When the front-end function is parameterized, 
maximizing MI with respect to the front-end parameters naturally yields 
channel architectures that maximally convey information 
from the input to the output under noise. In this section, 
we use the information gradient to maximize MI.

\subsubsection{Problem Setup}
We consider the linear vector channel with Gaussian input again:
\begin{equation}
\label{eq:channel}
Y \;=\; \bm A X +  Z, \qquad
X \sim \mathcal N(\bm 0, \sigma_x^2 \bm I_n),\quad
Z \sim \mathcal N(\bm 0, t \bm I_m),
\end{equation}
with a parameter matrix $\bm A \in \mathbb R^{m\times n}$ subject to a Frobenius-norm constraint
$
\|\bm A\|_F \le P,
$
where $P$ is a positive constant.
In this case, the mutual information is given by
\begin{equation}
\label{eq:I-logdet}
I(\bm A) = \frac{1}{2}\log\det\!\Big(\bm I_m + \frac{\sigma_x^2}{t}\,\bm A\bm A^\top\Big).
\end{equation}
Our goal in this subsection is the constrained maximization:
\begin{equation}
\label{eq:opt-A}
\text{maximize}_{\bm A}\; I(\bm A)\quad \text{s.t.}\quad \|\bm A\|_F\le P,
\end{equation}
by using the information gradient.

\subsubsection{Preparation}
Assume $\bm\Sigma_x=\sigma_x^2\bm I_n$ as in \eqref{eq:channel}. 
Then, we have
\begin{align}
I(\bm A)
&=\frac12\log\det\!\Big(\bm I_m+\frac{\sigma_x^2}{t}\,\bm A\bm A^\top\Big) \\
&=\frac12\sum_{i=1}^{m}\log\!\big(1+c\,s_i^2\big),\quad c\equiv\sigma_x^2/t,
\end{align}
where $s_1,\dots,s_m$ are the singular values of $\bm A$ (padding with zeros if $m>n$). The constraint is
$\sum_{i=1}^m s_i^2=\|\bm A\|_F^2\le P^2$.
Since $x\mapsto \log(1+c\,x)$ is concave and symmetric 
in the variables $\{s_i^2\}$, the sum is Schur-concave; by Jensen's inequality 
(or KKT conditions), the maximizer allocates energy uniformly:
\[
s_1^2=\cdots=s_m^2=\frac{P^2}{m}.
\]
Therefore the maximal mutual information equals
\begin{equation}
\label{eq:theory-optimum}
I^\star
=\frac{m}{2}\log\!\Big(1+\frac{\sigma_x^2}{t}\,\frac{P^2}{m}\Big).
\end{equation}


For solving the optimization problem \eqref{eq:opt-A}, 
we use the {\em projected gradient ascent method} based on the alternating optimization
described in Subsection III-B.

In a projected gradient ascent step, we perform the following steps:
given the current $\bm A^{(k)}$, 
\begin{enumerate}
\item Fit $s_{\bm \theta^{(k)}}$ on samples from \eqref{eq:channel} generated with $\bm A^{(k)}$ using DSM learning.
\item Form a Monte Carlo estimate of \eqref{eq:grad-vjp}:
\begin{equation}
\label{eq:grad-vjp}
\widehat{\nabla}_{\!\bm A} I(\bm A^{(k)}) \;=\; -\frac{1}{N}\sum_{i=1}^N s_{\bm \theta^{(k)}}(Y^{(i)})\,X^{(i)\top}.
\end{equation}
\item Update $\bm A^{(k+1)}$ with a \emph{fixed} 
step size $\gamma >0$ and project onto the Frobenius ball of radius $P$:
\begin{equation}
\label{eq:grad_update}
\bm A^{(k+1)} = \Pi_{\{\|\cdot\|_F\le P\}}\!\Big(\bm A^{(k)} 
+ \gamma \widehat{\nabla}_{\!\bm A} I(\bm A^{(k)})\Big),
\end{equation}
where $\Pi_{\{\|\cdot\|_F\le P\}}$ is the projection operator onto the Frobenius ball of radius $P$.
\end{enumerate}

\subsubsection{Experimental Setup}
We consider the linear vector channel with Gaussian input \eqref{eq:channel} 
with $n=m=8$, $\sigma_x^2=1$, $t=0.5$, 
and optimize $\bm A\in\mathbb R^{m\times n}$ 
under the Frobenius constraint $\|\bm A\|_F\le P$ with $P=5.0$. 
We run the alternating projected ascent with an \emph{unconditional} score model 
$s_{\bm \theta}:\mathbb R^m\to\mathbb R^m$ (input $\bm y$ only), 
trained by DSM at a \emph{fixed} noise level 
$\sigma_{\rm train}=\sigma_{\rm eval}=0.1\sqrt t$. 
At each outer iteration $k$, we train $s_{\bm \theta^{(k)}}$ for $1000$ 
epochs (two-layer MLP, width $256$, SiLU, batch size $4096$, AdamW with gradient clipping), 
estimate the VJP gradient with $N=5\times 10^4$ 
Monte Carlo samples and a Stein calibration 
factor applied at evaluation, then the gradient update
\eqref{eq:grad_update} is carried out.

The mutual information $I(\bm A^{(k)})$ used for reporting is computed from the log-det form \eqref{eq:I-logdet}. As a reference, we plot the theoretical optimum $I^\star$ in \eqref{eq:theory-optimum} as a horizontal line.

\subsubsection{Results}
\begin{figure}[t]
  \centering
  \includegraphics[width=0.9\linewidth]{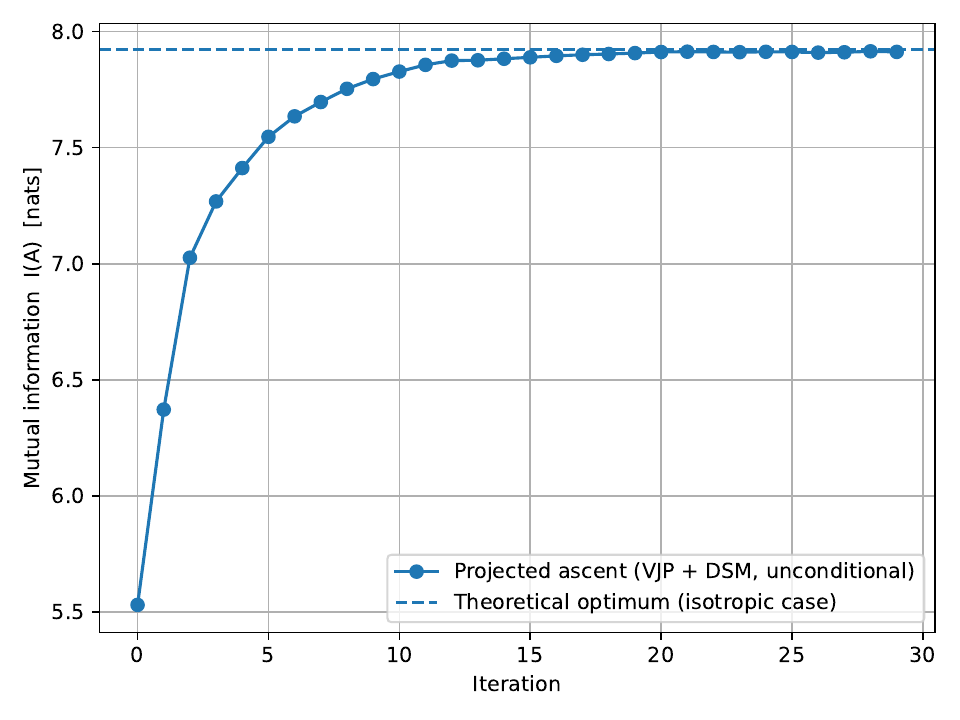}
  \vspace{-0.6em}
  \caption{Projected gradient ascent with unconditional DSM (fixed step).
  Iteration vs. mutual information $I(\bm A^{(k)})$ for $n{=}m{=}8$, $\sigma_x^2{=}1$, $t{=}0.5$, $\|\bm A\|_F\le 5.0$.
  The dashed line is the theoretical optimum $I^\star$ under the Frobenius constraint.}
  \label{fig:mi_iter_curve}
\end{figure}

Figure~\ref{fig:mi_iter_curve} shows a steady increase of $I(\bm A^{(k)})$ toward the theoretical optimum $I^\star$, 
confirming that (i) the VJP estimator coupled with 
an unconditional DSM score yields a usable ascent direction, 
and (ii) the projected update effectively respects the Frobenius
 constraint while improving information. 
 Minor fluctuations stem from Monte Carlo noise 
 and the finite-epoch training of the score model; they can be reduced by increasing $N$ 
 or epochs.
 The result supports the practicality of the proposed score-based information 
 gradients for mutual information maximization.

 \subsection{MI Maximization for Nonlinear Channels}

 \subsubsection{Experimental Setup}

 We consider the nonlinear channel
 \[
 Y=\tanh(\bm AX)+Z,\qquad X\sim\mathcal{N}(0,\bm I_n),
 \quad Z\sim\mathcal{N}(0,t \bm I_n),
 \]
 with dimension fixed to $n=12$. 
 The design matrix is constrained by the Frobenius norm
$\|\bm A\|_{\mathrm F}\le P$ with $P=5.0$.
We initialize $\bm A$ by a random Gaussian matrix scaled 
 to satisfy $\|\bm A\|_{\mathrm F}=P$.

 To stabilize the numerical evaluation of $I(X;Y)$, we estimate 
 the output entropy $H(Y)$ using a leave-one-out (LOO) Gaussian KDE 
 in a whitened space \cite{moon1995}.
  The KDE bandwidth is chosen by maximizing the LOO log-likelihood 
  over a small grid around Scott's rule-of-thumb.
 
 Noise variance $t$, the Frobenius-radius $P$, batch sizes, seeds, 
 early-stopping criteria follow the previousl experiment without 
 change and are omitted here for brevity. As the theoretical 
 optimum for this nonlinear channel is unknown, 
 we report only the estimated MI.
 
 \subsubsection{Results}
 Figure~\ref{fig:mi_tanhA} shows the estimated MI trajectory 
 obtained by gradient ascent processes. Although small fluctuations 
 are visible across iterations, the overall trend is clearly increasing, 
 indicating that the projected gradient-ascent with the information 
 gradient successfully drives the parameter updates in the correct (MI-maximizing) 
 direction. This behavior is consistent with our expectation 
 that the nonlinearity $\tanh(\cdot)$ does not prevent ascent 
 when gradients are estimated stably, while the KDE-based MI evaluation provides 
 a robust external metric for progress monitoring. 

 \begin{figure}[t]
   \centering
   \includegraphics[width=0.9\linewidth]{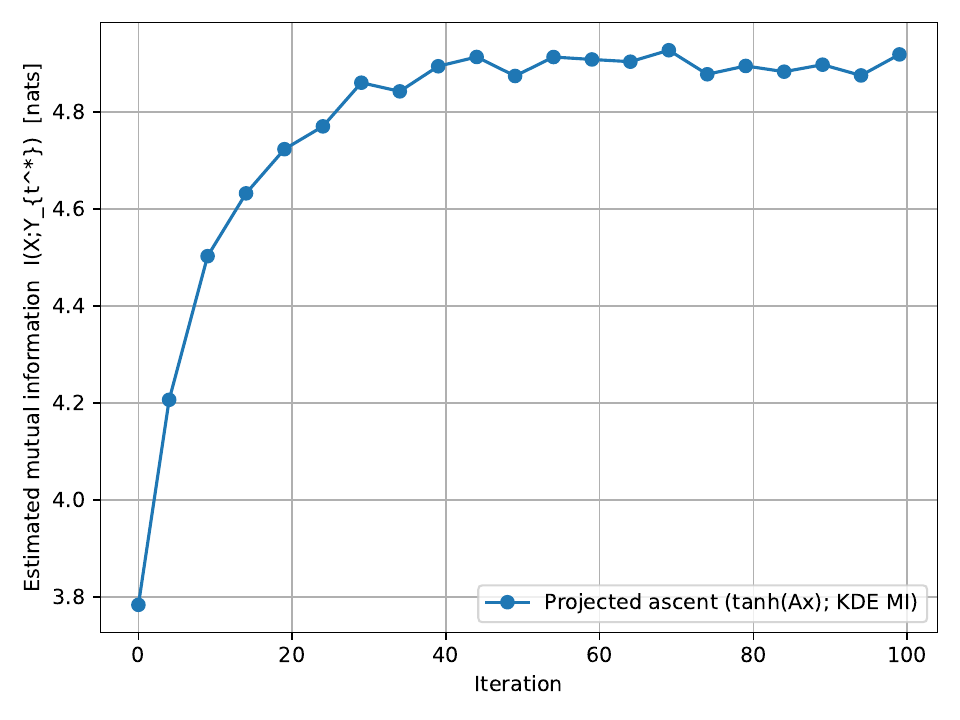}
   \caption{MI estimate versus iteration for $Y=\tanh(AX)+Z$ ($n{=}12$).}
   \label{fig:mi_tanhA}
 \end{figure}

 \subsection{Task-Oriented IB Optimization in Linear Vector Channel with Gaussian Input}

\subsubsection{Experimental setup}
We consider a linear vector channel with Gaussian input and a linear task variable:
\begin{align}
  Y &= \bm A X +  Z,\qquad  Z \sim \mathcal{N}(0, t \bm I_m),\\
  T &= \bm W X,\qquad  X \sim \mathcal{N}(0,\bm \Sigma_x),
\end{align}
where $\bm A\in\mathbb{R}^{m\times n}$ is the channel matrix subject to a Frobenius-norm constraint $\|\bm A\|_F \le P$, 
and $\bm W\in\mathbb{R}^{k\times n}$ defines the task.
The parameter $t>0$ denotes the noise variance. 
The objective is the IB criterion:
\begin{align}
  L_{\mathrm{IB}}(A) = I(T;Y) - \beta I(X;Y),
\end{align}
with trade-off parameter $\beta>0$. We maximize $L_{\mathrm{IB}}$ over $\bm A$ by projected gradient ascent.

For each iterate $\bm A$, mutual informations are evaluated exactly in closed form:
\begin{align}
  I(X;Y) &= \frac12 \log\det\!\Big(\bm I_m + \tfrac{1}{t}\, \bm A \bm \Sigma_x \bm A^\top\Big),\\
  I(T;Y) &= \frac12 \log\det\!\Big(\bm I_k + \bm \Sigma_T^{-1}\, \bm \Sigma_{TY}\, \bm \Sigma_Y^{-1}\, \bm \Sigma_{YT}\Big),
\end{align}
where $\bm \Sigma_T = \bm W\bm \Sigma_x \bm W^\top$, $\bm \Sigma_Y = \bm A\bm \Sigma_x \bm A^\top + t \bm I_m$, and $\bm \Sigma_{TY} = \bm W\bm \Sigma_x \bm A^\top$.

Let $s_Y(\bm y) = \nabla_{\bm y} \log p_Y(\bm y)$ and $s_{Y|T}(\bm y|\bm t) = \nabla_{\bm y} \log p_{Y|T}(\bm y|\bm t)$.
By the IB gradient identity, an ascent direction is given by
\begin{align}           
  \nabla_A L_{\mathrm{IB}} \propto
  \mathbb{E}\!\big[\underbrace{s_{Y|T}(Y\!\mid T)}_{\text{(cond.)}}+(\beta-1)\underbrace{s_Y(Y)}_{\text{(uncond.)}}\big]\, X^\top.
\end{align}    
In this experiment we \emph{learn both scores} by DSM learning.

The dimension is set to $n=m=12$.
Unless otherwise noted, we fix the input covariance to 
$\bm \Sigma_x=\bm I_{12}$ and draw the task matrix $\bm W\in\mathbb{R}^{k\times 12}$ once at random (row full-rank).
In our runs we set $k=4$, noise variance $t=0.5$, trade-off $\beta=1.0$, and impose the Frobenius constraint $\|A\|_F\le P$ with $P=5.0$.
The DSM perturbation level is $\sigma=0.1$ for both unconditional and conditional score learners.
At each outer iteration, we perform $N_{\mathrm{DSM}}=200$ DSM updates for each score network with Adam (learning rate $10^{-3}$) on mini-batches of size $B=512$.
The ascent step size for $A$ is $\gamma=0.05$, followed by projection onto the Frobenius ball.

\subsubsection{Results}
Figure~\ref{fig:ib_both_dsm_curve_n12} reports the three tracked metrics per iteration:
the IB objective $L_{\mathrm{IB}}=I(T;Y)-\beta I(X;Y)$ and the two mutual informations $I(T;Y)$ and $I(X;Y)$ (all in nats), 
each evaluated in closed form at the current $\bm A$.
As predicted by the IB-gradient ascent, $L_{\mathrm{IB}}$ exhibits a consistent increase and stabilizes after a transient phase.
Concurrently, $I(T;Y)$ grows as the channel aligns with task-relevant directions, while $I(X;Y)$ is suppressed by the $\beta$-penalty and the Frobenius constraint, resulting in a characteristic trade-off curve over iterations.

This result provides strong empirical validation for our proposed IB 
gradient (Proposition~\ref{prop:ib_gradient}). It demonstrates 
that the VJP-based gradient estimator, using \emph{only} learned scores for both 
$s_Y$ and $s_{Y|T}$, provides a practical and effective ascent direction. 
This is clearly evidenced by the optimization successfully navigating the 
IB trade-off: maximizing the objective $\mathcal{L}_{IB}$ by preserving 
task-relevant information $I(T;Y)$ while actively compressing the total information $I(X;Y)$.

\begin{figure}[t]
  \centering
  \includegraphics[width=0.9\linewidth]{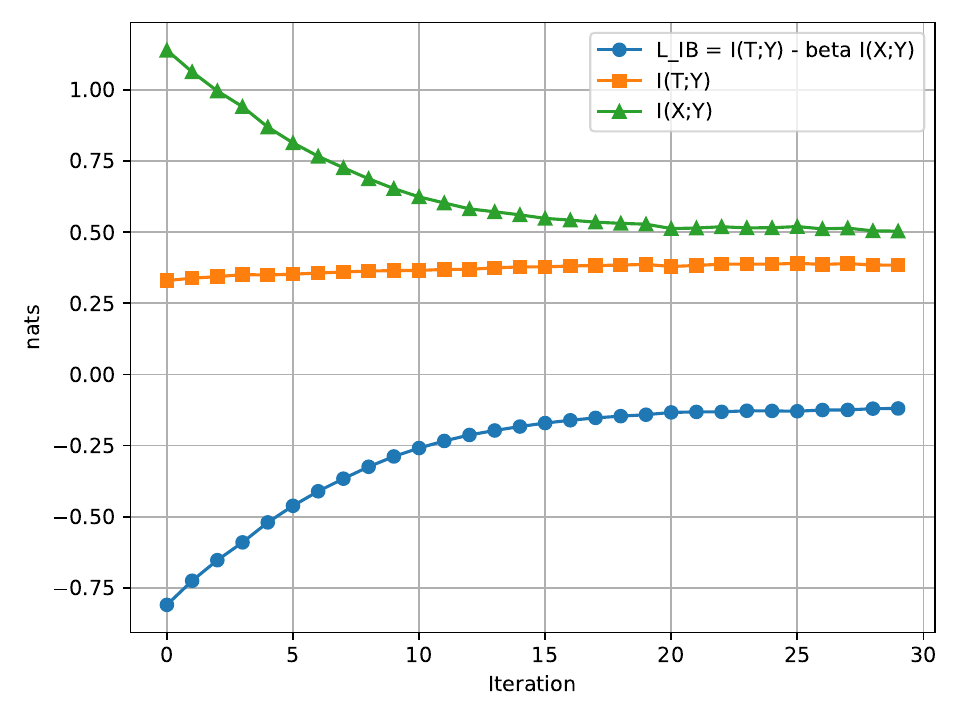}
  \caption{Mutual information in task-oriented IB optimization with $n=m=12$ and $k=4$.
  Both unconditional and conditional scores are learned by DSM; $I(T;Y)$ and $I(X;Y)$ are computed in closed form at each iteration.
  Curves show $L_{\mathrm{IB}}$, $I(T;Y)$, and $I(X;Y)$ (nats) 
  versus iteration under the Frobenius constraint $\|\bm A\|_F\le 5.0$.
  }
  \label{fig:ib_both_dsm_curve_n12}
\end{figure}

\section{Conclusion}
Leveraging the score-to-Fisher bridge methodology \cite{wadayama2025}, 
we proposed a tractable and efficient framework for optimizing parametric nonlinear Gaussian channels 
by maximizing information-theoretic objectives. The core of our contribution is a 
novel formula for the information gradient, $\nabla_{\eta}I(X;Y_t)$, 
derived by using the marginal score function $s_{Y_t}$. 
This approach avoids intractable marginal computations and relies only on 
DSM learning and forward channel sampling.
We extended this gradient formula to task-oriented objectives 
$\nabla_{\eta}I(T;Y_t)$ (Proposition~\ref{prop:task_gradient}) and, 
more generally, to the IB objective 
$\nabla_{\eta}\mathcal{L}_{IB}$ (Proposition~\ref{prop:ib_gradient}), 
where $\mathcal{L}_{IB} = I(T;Y_t) - \beta I(X;Y_t)$. 
The resulting gradient estimators are implemented efficiently 
using the VJP identity, enabling practical 
end-to-end optimization via alternating gradient ascent.
Future directions include multi-terminal extensions, 
integration with finite-length analysis, rate-distortion theory,
and applications to emerging 
6G communication paradigms such as semantic and goal-oriented communications \cite{Gunduz2023}.

\section*{Acknowledgments}
This work was supported by JST, CRONOS, Japan Grant Number JPMJCS25N5.






\begin{thebibliography}{00}

\bibitem{moon1995}
Y.-I.~Moon, B.~Rajagopalan, and U.~Lall,
``Estimation of mutual information using kernel density estimators,''
\textit{Physical Review E}, vol.~52, no.~3, pp.~2318--2321, 1995.
doi: \href{https://doi.org/10.1103/PhysRevE.52.2318.}{10.1103/PhysRevE.52.2318.}
  

\bibitem{wadayama2025}
T.~Wadayama, 
``Mutual information estimation via score-to-Fisher bridge
for nonlinear Gaussian noise channels,''\\
\url{https://arxiv.org/abs/2510.05496}, 2025.

\bibitem{hyvarinen2005}
A.~Hyv\"{a}rinen,
``Estimation of non-normalized statistical models by score matching,''
\textit{Journal of Machine Learning Research}, vol.~6, pp.~695--709, 2005.\\
Available: \url{https://jmlr.org/papers/volume6/hyvarinen05a/hyvarinen05a.pdf}.

\bibitem{vincent2011}
P.~Vincent,
``A connection between score matching and denoising autoencoders,''
\textit{Neural Computation}, vol.~23, no.~7, pp.~1661--1674, Jul.~2011.
doi: \href{https://doi.org/10.1162/NECO\_a\_00142.}{10.1162/NECO\_a\_00142.}

\bibitem{song2019}
Y.~Song and S.~Ermon,
"Generative modeling by estimating gradients of the data distribution,''
in \textit{Advances in Neural Information Processing Systems (NeurIPS)}, 2019, pp.~11895--11907.
arXiv:1907.05600.

\bibitem{song2021}
Y.~Song, J.~Sohl-Dickstein, D.~P.~Kingma, A.~Kumar, S.~Ermon, and B.~Poole,
``Score-based generative modeling through stochastic differential equations,''
in \textit{Advances in Neural Information Processing Systems (NeurIPS)}, vol.~34, 2021, pp.~30470--30480.
arXiv:2011.13456.

\bibitem{Baydin2018}
A.~G. Baydin, B.~A. Pearlmutter, A.~A. Radul,  and J.~M. Siskind, ``Automatic
  differentiation in machine learning: a survey,'' \emph{Journal of Machine
  Learning Research}, vol.~18, pp. 5595--5637, 2018.

\bibitem{Gunduz2023}
D.~G{\"u}nd{\"u}z, Z.~Qin, I.{\~n}aki~Estella~Aguerri, H.~S.~Dhillon, Z.~Yang, A.~Yener, K.-K.~Wong, and C.-B.~Chae,
``Beyond transmitting bits: context, semantics, and task-oriented communications,''
\emph{IEEE Journal on Selected Areas in Communications}, vol.~41, no.~1, pp.~5--41, Jan.~2023., 
doi: 10.1109/JSAC.2022.3223408.

\bibitem{Tishby1999}
N.~Tishby, F.~C.~Pereira, and W.~Bialek, ``The information bottleneck method,''
in \emph{Proc. 37th Annu. Allerton Conf. Commun., Control, and Computing},
Monticello, IL, USA, 1999, pp.~368--377.

\bibitem{fisher1922}
R.~A.~Fisher,
``On the Mathematical Foundations of Theoretical Statistics,''
\textit{Philosophical Transactions of the Royal Society of London. Series A},
vol.~222, pp.~309--368, 1922.
doi: \href{https://doi.org/10.1098/rsta.1922.0009.}{10.1098/rsta.1922.0009.}

\bibitem{cover2006}
T.~M.~Cover and J.~A.~Thomas,
\textit{Elements of Information Theory}, 2nd~ed.
Hoboken, NJ, USA: Wiley, 2006. ISBN: 978-0-471-24195-9.


\bibitem{stam1959}
A.~J.~Stam,
``Some inequalities satisfied by the quantities of information of Fisher and Shannon,''
\textit{Information and Control}, vol.~2, no.~2, pp.~101--112, Jun.~1959.
doi: \href{https://doi.org/10.1016/S0019-9958(59)90348-1}{10.1016/S0019-9958(59)90348-1.}

\bibitem{barron1986entropy}
A.~R.~Barron,
``Entropy and the central limit theorem,''
\textit{The Annals of Probability}, vol.~14, no.~1, pp.~336--342, Jan.~1986.
doi: \href{https://doi.org/10.1214/aop/1176992632.}{10.1214/aop/1176992632.}



\end{thebibliography}
\end{document}